\documentclass[
 reprint, 
 superscriptaddress,
 amsmath,amssymb,
 aps, pra,
 floatfix,
 showkeys
]{revtex4-2}

\UseRawInputEncoding

\usepackage{hyperref}
\usepackage{physics}
\usepackage[caption=false]{subfig}
\usepackage{amssymb}
\usepackage{amsthm}
\usepackage{import}
\usepackage{pgfplots}
\usepackage{orcidlink}
\usepackage{ragged2e}
\usepackage{silence}
\usepackage{xstring}
\usepackage{iftex}
\usepackage{orcidlink}
\usepackage{academicons}
\usepackage{xcolor}
\definecolor{orcidlogocol}{HTML}{A6CE39}

\pgfplotsset{compat=1.14}

\newcommand{\tnorm}[1]{\big \| #1 \big \|}
\newcommand{\mathdefault}[1][]{}

\newtheorem{namedtheorem}{Proposition}
\newtheorem{proposition}{Proposition}
\newtheorem{lemma}{Lemma}

\newcommand{\blue}[1]{{#1}}

\renewcommand{\selectlanguage}[1]{}
\begin{document}

\preprint{APS/123-QED}

\title{\textbf{Quantum Dissipative Search via Lindbladians}}

\author{Peter J. Eder\textsuperscript{\orcidlink{0009-0006-3244-875X}}}
\email{peter-josef.eder@siemens.com}
\affiliation{Technical University of Munich, School of CIT, Department of Computer Science, Garching, Germany}
\affiliation{Siemens AG, Munich, Germany}

\author{Jernej Rudi Fin\v{z}gar\textsuperscript{\orcidlink{0000-0003-4393-827X}}}
\email{jernej-rudi.finzgar@tum.de}
\affiliation{Technical University of Munich, School of CIT, Department of Computer Science, Garching, Germany}
\affiliation{BMW AG, Munich, Germany}

\author{Sarah Braun\textsuperscript{\orcidlink{0000-0002-7032-6116}}}
\email{sarah.braun@siemens.com}
\affiliation{Siemens AG, Munich, Germany}

\author{Christian B. Mendl\textsuperscript{\orcidlink{0000-0002-6386-0230}}}
\email{christian.mendl@tum.de}
\affiliation{Technical University of Munich, School of CIT, Department of Computer Science, Garching, Germany}
\affiliation{Technical University of Munich, Institute for Advanced Study, Garching, Germany}

\date{\today}

\begin{abstract}
Closed quantum systems follow a unitary time evolution that can be simulated on quantum computers. By incorporating non-unitary effects via, e.g., measurements on ancilla qubits, these algorithms can be extended to open-system dynamics, such as Markovian processes described by the Lindblad master equation. In this paper, we analyze the convergence criteria and speed of a Markovian, purely dissipative quantum random walk on an unstructured classical search space. Notably, we show that certain jump operators make the quantum process replicate a classical one, while others yield differences between open quantum (OQRW) and classical random walks. We also clarify a previously observed quadratic speedup, demonstrating that OQRWs are no more efficient than classical search. Finally, we analyze a dissipative discrete-time ground-state preparation algorithm with a lower implementation cost. This allows us to interpolate between the dissipative and the unitary domain and thereby illustrate the important role of coherence for the quadratic speedup.
\end{abstract}

\maketitle

\section{Introduction}

\begin{sloppypar}
There has been rising interest in the simulation of Markovian (i.e., memoryless) open-system dynamics following the Lindblad Master equation (LME)~\cite{Gorini1976, Lindblad1976, Manzano2020} on quantum computers in recent years~\cite{Chen:2023cuc, ding2024simulating, chen2023efficient, cleve_et_al:LIPIcs.ICALP.2017.17, li2023simulating}. Potential applications include state preparation~\cite{Kraus2008, Verstraete2009, Kastoryano2016}, the simulation of realistic quantum systems~\cite{PhysRevD.106.054508, Tornow2022, Rall_2023, Dorner2012, RevModPhys.59.1, May2011}, or the study of noise in quantum computing devices~\cite{Temme2014, Shenvi2003}. Ref.~\cite{Ding2023-rv} focuses on an early fault-tolerant implementation for preparing ground states via Lindbladians. Unlike previous methods, the scheme does not rely on cumbersome block encodings or precise control circuits, and requires only a single ancilla qubit, making it suitable for the early fault-tolerant regime with constrained quantum resources. A major difference compared to alternative approaches, such as employing Quantum Singular Value Transform~\cite{PRXQuantum.3.040305} (QSVT), is that it does not assume an easy-to-prepare initial state $\rho_0$ with overlap ${p_g = \bra{\Psi_g}\rho_0\ket{\Psi_g}= \Omega\left(1/\mathrm{poly}(n)\right)}$ with the ground state~$\ket{\Psi_g}$, where $n$ denotes the number of qubits in the system. Typically, ground-state preparation algorithms rely on assumptions such as a substantial initial overlap to justify their efficiency. From the perspective of complexity theory, preparing the ground state of a few-body Hamiltonian is QMA-hard~\cite{aharonov2002quantum, Nielsen2012}.

Given a Hamiltonian written in its eigenbasis ${H = \sum_{i=0}^N \lambda_i \ket{i}\bra{i}}$ and the density matrix $\rho$ of the system, the LME can be written as:
\begin{equation}
\label{eqn:original_lindblad_equation}
    \frac{\textrm{d}\rho}{\textrm{d}t} = \underbrace{-i[H,\rho]}_{\mathcal{L}_H[\rho]} + \underbrace{\sum_{j} \left(L_j \rho L_j^\dag - \frac{1}{2} \left\{L_j^\dag L_j, \rho\right\}\right)}_{\mathcal{L}_L[\rho]} =: \mathcal{L}[\rho].
\end{equation}
The operators $L_j$, often referred to as jump operators, represent the channels through which the system interacts with the environment~\cite{Breuer2007}. According to Ref.~\cite{lindblad_superoperator}, the equation can be rewritten in the Fock-Liouville space using a matrix of dimension $N^2 \times N^2$ ($N=2^n$), which is referred to as Liouvillian or Lindbladian, consisting of a unitary and a dissipative part, $\mathcal{L} := \mathcal{L}_H + \mathcal{L}_L$. Intuitively, the LME~\ref{eqn:original_lindblad_equation} performs a continuous-time quantum random walk (QRW) on the energy spectrum of $H$~\cite{Chen:2023cuc}.

Previous work studied the dynamics of discrete-time open quantum random walks (OQRWs) on a one-dimensional chain, wherein the quantum walker is driven exclusively by the interaction with the environment~\cite{Attal2012, Sinayskiy2013}. The formalism is built on completely positive and trace preserving (CPTP) maps, such as time evolution governed by the LME. Ref.~\cite{Attal2012} shows that OQRWs can be recovered from \blue{unitary quantum random walks (UQRWs)} if the transition operators satisfy an extended normalization condition by working in an enlarged Hilbert space, which is eventually traced out. Conversely, obtaining UQRWs from OQRWs appears feasible only if the decoherence operations are excluded. In contrast to the unusual behavior of UQRWs, it was found that OQRWs converge to a Gaussian distribution under certain conditions~\cite{Attal2014, Konno2012}. Ref.~\cite{Attal2012} further demonstrates that OQRWs include the classical random walk (CRW) as a special case. Lastly, Ref.~\cite{Pellegrini2014} derives a natural continuous-time extension of OQRWs~\cite{PhysRevA.75.022103}.

Few analyses exist on purely dissipative Markovian dynamics in the context of the Grover problem. Ref.~\cite{Allahverdyan2022} investigates the classical version of the dissipative part of the LME, which is also known as the rate equation. The authors demonstrate that a shift in the energy levels of the eigenstates of the Grover Hamiltonian that is not larger than the spectral gap leads to convergence to the ground state in ${\mathcal{O}(\log N)}$ steps for bounded long-range transitions. We will show later that a classical greedy algorithm can achieve an equivalent complexity. Note that the transition regimes such as long- and short-range, which are depicted in Figure~\ref{fig:transition_regimes}, control the connectivity of the random walk. Ref.~\cite{Vogl_2010} analyzes the purely dissipative component $\mathcal{L}_L$ of the quantum LME and derives the convergence speed to the ground state for long-range transitions, finding a quadratic speedup compared to the classical rate equation. However, this speedup still results in a complexity of $\mathcal{O}(N)$. The authors argue that this is due to the stringent restrictions on the transition rates. We will clarify this apparent speedup by demonstrating the transition from the LME to the rate equation. Lastly, Ref.~\cite{Ding2023-rv}, although not specifically addressing the Grover problem, provides an ansatz for ground-state preparation via Lindbladians and proves that the mixing time is polynomial in the system size $n$ under certain assumptions.

In this paper, we partition these ansätze into different coupling regimes that control the random walk, and highlight the parallels and differences between the classical and dissipative approaches. Furthermore, we interpret the results in the context of an actual implementation on a digital quantum device to properly evaluate the resource requirements of the algorithm. Finally, we analyze a dissipative discrete-time ground-state preparation algorithm with a lower implementation cost, which deviates from the LME but still provably prepares the ground state provided that the ansatz is suitable. This analysis reveals how certain coupling regimes maintain Lindblad dynamics, while others differ. It further allows us to interpolate between the dissipative and the unitary domain and thereby illustrate the critical role of coherence for the quadratic Grover speedup, a topic that has been explored extensively in previous studies~\cite{Sun2024, Rastegin2018_1, Rastegin2018_2}. Our contributions are as follows:
\begin{itemize}
    \item We examine the aforementioned early fault-tolerant algorithm for ground-state preparation via Lindbladians~\cite{Ding2023-rv} in detail, offering a clear and intuitive exposition (Section~\ref{sct:Single-Ancilla Ground-State Preparation via Lindbladians} and Appendix~\ref{sct:Single-Ancilla Ground-State Preparation via Lindbladians_Appendix}). Going beyond existing work, we refine the complexity of implementation for the Grover problem (Proposition~\ref{prop:Linear_scaling}) and outline the advantages of including the unitary dynamics via $\mathcal{L}_H$ (Appendix~\ref{sct:Influence of the Hamiltonian}).
    \item By closely aligning with an actual implementation, we can accurately interpret our findings and provide strong evidence that Markovian open-system dynamics for finding marked elements in an unstructured classical search space are no more effective than classical search. In doing so, we clarify a previously reported quadratic speedup~\cite{Vogl_2010} (Section~\ref{sct:Results for Grover's Search}).
    \item We further present an ansatz for which the continuous-time OQRW replicates a CRW, and show that in two other scenarios, the dynamics differ (Section~\ref{sct:Results for Grover's Search}). Finally, we highlight the essential role of coherence in achieving the quadratic speedup (Section~\ref{sct:Discrete-Time Random Walk}) by examining the discrete-time ground-state preparation algorithm from \mbox{Ref.~\cite[Appendix B.5]{Ding2023-rv}}, with an improved implementation cost, and identify a pathway to transition to the UQRW (Section~\ref{sct:Discrete-Time Random Walk}).
\end{itemize}

In the following, we denote the spectral norm of a matrix $A$, which is equal to its maximum singular value, as $\norm{A}$. It is also referred to as operator norm. Additionally, the trace norm of $A$ is given by $\norm{A}_1=\Tr \left(\sqrt{A^\dag A}\right)$.

\section{Single-Ancilla Ground-State Preparation via Lindbladians}
\label{sct:Single-Ancilla Ground-State Preparation via Lindbladians}

In this section we will briefly outline the quantum algorithm introduced in Ref.~\cite{Ding2023-rv} to prepare ground states via Markovian open-system dynamics on an early fault-tolerant digital quantum computer. We will sketch the general concept, our improvement to the implementation and the insights that we later use to interpret our findings. A more detailed and intuitive description of the algorithm is provided in Appendix~\ref{sct:Single-Ancilla Ground-State Preparation via Lindbladians_Appendix}.

If we discard the generator of unitary evolution $\mathcal{L}_H$ and employ a single jump operator $L$, Eq.~\ref{eqn:original_lindblad_equation} reduces to the purely dissipative dynamics
\begin{equation}
\label{eqn:ground_state_lindblad_equation}
    \frac{\textrm{d}\rho}{\textrm{d}t} = \mathcal{L}_L[\rho] = L \rho L^\dag - \frac{1}{2} \{L^\dag L, \rho\},
\end{equation}
which admit a formal solution
\begin{equation}
\label{eqn:formal_solution_lindblad_equation}
    \rho(t)=e^{\mathcal{L}t}[\rho(0)].
\end{equation}
From Eq.~\ref{eqn:formal_solution_lindblad_equation} we can see that in order to model the LME~\ref{eqn:ground_state_lindblad_equation} on a digital quantum device, we need to simulate the action of the non-unitary operator ${e^{\mathcal{L}t}}$. Then, to prepare the ground state of a Hamiltonian $H$, we must find a proper ansatz for $L$ which achieves ergodicity (i.e., the Markov chain is irreducible and aperiodic) and satisfies the quantum detailed balance condition~\cite{Chen:2023cuc}. These are sufficient conditions for the convergence of a Markov chain to a stationary distribution~\cite{Levin2017-na}. Ref.~\cite{Ding2023-rv} shows that they can be satisfied for certain Hamiltonians using a single jump operator $L$. We will show later that this is also the case for the Grover problem. This enables the design of a simulation scheme utilizing just one ancilla qubit.
We introduce $L$ as
\begin{equation}
    \label{eqn:jumperator_L_definition}
    L := \sum_{i,j \in [N]} \hat{\gamma}(\lambda_i - \lambda_j) \ket{i} \bra{i} A \ket{j} \bra{j} = \int_{-\infty}^{\infty} \gamma(s) A(s) \textrm{d}s.
\end{equation}
The Hermitian operator ${A(s)}$ is usually referred to as the \emph{coupling operator} and is chosen together with the Hamiltonian $H$ and the \emph{filter function} ${\gamma(s)}$ as hyperparameters of the algorithm. The coupling operator follows a Heisenberg evolution ${A(s) = e^{iHs} A e^{-iHs}}$ and the filter function is defined via the inverse Fourier transform ${\gamma(s):= \frac{1}{2\pi}\int_{\mathbb{R}} \hat{\gamma}(\omega) e^{-i\omega s}\textrm{d}\omega}$. Intuitively, the Lindbladian performs a continuous-time random walk on the energy spectrum of $H$, driven by the coupling operator $A$ that determines the transitions. $A$ is written in the energy eigenbasis (i.e., $A_{ij}$ represents the transition strength between energy eigenstates $\ket{\Psi_i}$ and $\ket{\Psi_j}$) and weighted by a bath-dependent function $\hat{\gamma}(\omega)$, which is real and non-negative~\cite{Chen:2023cuc}.

Three foundational building blocks are required to simulate the action of the non-unitary operation ${e^{\mathcal{L} t}[\rho]}$ from Eq.~\ref{eqn:formal_solution_lindblad_equation}~\cite{Ding2023-rv}. First, we have to find an approximation of ${e^{\mathcal{L} t}}$ that can be directly implemented on a quantum device. We can achieve this by repeatedly combining a unitary time evolution generated by a Hermitian, dilated jump operator $\tilde{L}$, which block encodes $L$, with non-unitary measurements on an ancilla qubit, initially prepared in the state $\ket{0}_a$. Define $\tilde{L}$ as
\begin{equation}
\label{eqn:L_tilde_definition}
    \tilde{L} := \begin{pmatrix} 0 & L^\dag \\ L & 0 \end{pmatrix}.
\end{equation}
Further, let ${\Tr_a \left( \sum_{i,j = 0}^1 \ket{i}_a \bra{j}_a \otimes \rho_{ij}\right) = \sum_{i=0}^1 \rho_{ii}}$ denote the operation that traces out the subspace of the ancilla qubit, which is used for the block-encoding. The \mbox{$m$-th} step of the approximation scheme is then given by:
\begin{equation}
\label{eqn:sigma_lindblad_approximation}
\begin{aligned}
    \sigma(m\tau)&:=\sigma_m(\tau) \\
    &:= \Tr_a \left( e^{-i\tilde{L}\sqrt{\tau}} [\ket{0}_a \bra{0}_a \otimes \sigma_{m-1}(\tau)] e^{i \tilde{L} \sqrt{\tau}}\right),
\end{aligned}
\end{equation}
where the time step $\tau$ and the total number $\mathcal{N}$ of iterations have to be chosen appropriately, such that ${\norm{\rho(\mathcal{N}\tau) - \sigma_{\mathcal{N}}(\tau)}_1 = \mathcal{O}(\epsilon)}$ for some desired accuracy $\epsilon>0$. Here, $\rho(t)$ is defined via the LME~\ref{eqn:ground_state_lindblad_equation}. We elaborate on the specific choice of these parameters in Appendix~\ref{sct:Error Analysis}. Note that after each measurement operation, the ancilla needs to be reset to the state $\ket{0}_a$, and its outcome can be discarded.

Next, to avoid a direct and expensive implementation of the time evolution operator ${e^{-i\tilde{L}\sqrt{\tau}}}$, one has to find a filter function $\gamma(s)$ that decays rapidly as ${\abs{s} \rightarrow \infty}$. This way, we can truncate and discretize the integral from Eq.~\ref{eqn:jumperator_L_definition}. We will see later in Appendix~\ref{sct:Trotter Splitting} that implementing the truncated integral results in much simpler circuits compared to those required for a block-encoded jump operator $\tilde{L}$. The discretization can for instance be achieved by means of the trapezoidal rule, as demonstrated in Appendix~\ref{sct:Truncation and Discretization of the Integral}.
Finally, the unitary evolution under the approximated dilated jump operator will contain a sum of Hermitian operators, where each constituent contributes to the discretized integral. The last step therefore consists of using the Trotter splitting of this exponential (see Appendix~\ref{sct:Trotter Splitting}).

The overall cost of implementation needed to approximate the LME up to a desired accuracy of $\epsilon>0$, which we provide in terms of the total Hamiltonian simulation time $T_H$, scales as (see Appendix~\ref{sct:Error Analysis}):
\begin{align}
\label{eqn:continuous_time_overall_complexity_lindblad}
\begin{aligned}
    T_H&=\mathcal{O}\left(\mathcal{N} \sqrt{\tau}\right) \\
    &=\tilde{\Theta}\left(\max\left(\norm{L}^4 T^2 \epsilon^{-1}, \norm{A}^{2+o(1)} T^{2+o(1)} \epsilon^{-o(1)}\right)\right).
\end{aligned}
\end{align}
Here, $\mathcal{N}$ is the number of times that the ancilla qubit is traced out and reset to the state $\ket{0}_a$ and $T$ stands for the simulation time. Note that there is a fundamental difference between $T$ and $T_H$: $T$ is the duration for which the LME is simulated, whereas $T_H$ corresponds to the overall time that terms of the form $e^{\pm iHt}$ are implemented, which are needed for the approximation scheme (see Appendix~\ref{sct:Trotter Splitting}). An increase in $\norm{A}$, will lead to a reduction in $T$; however, $T_H$, which is equivalent to the overall complexity of implementation, will remain constant. The $o(1)$ terms arise from employing a high-order Trotter formula, which results in an exponent of the form ${o(1):=1/p}$ as ${p\rightarrow \infty}$. Further, we introduce the notation $\tilde{\Theta}(g):=\Theta\left(g\,\mathrm{polylog}(g)\right)$.

We can see from Eq.~\ref{eqn:continuous_time_overall_complexity_lindblad} that the cost of implementation scales quadratically in $T$. However, we will show that the complexity can be reduced to be linear in certain cases. In the following proposition, we will prove this statement using the Grover problem as an example.

\begin{proposition}
\label{prop:Linear_scaling}
    To find the ground state of an unstructured search space with $M$ marked elements, the overall complexity of implementation in terms of the total Hamiltonian simulation time $T_H$ of the approximated continuous-time dynamics scales linearly in the mixing time. This linear scaling is achieved by choosing the coupling operator as $A:=J_N/\sqrt{(N-M)M}$, where $J_N$ denotes the all-one matrix of size $N\times N$ and the filter function $\hat{\gamma}(\omega)$ is defined to be zero for $\omega\geq0$ and one otherwise.
\end{proposition}

We defer the proof to Appendix~\ref{sct:Error Analysis} after presenting the algorithm in more detail. The basic idea is to show that by increasing the magnitude of $A$, the time to reach a steady state will decrease. Hence, we can shift all complexity into the norm of $A$ such that $T=1$. We demonstrate that in this scenario, $\norm{L}=1$, which makes it possible to reduce the first argument in Eq.~\ref{eqn:continuous_time_overall_complexity_lindblad} as ${\mathcal{O}\left(\norm{L}^4 T^2 \epsilon^{-1}\right) \rightarrow \mathcal{O}\left(\epsilon^{-1}\right)}$. This leads to an overall complexity of the form ${\tilde{\Theta}\left(\mathrm{max}\left(\epsilon^{-1},\norm{A}^{2+o(1)}\epsilon^{-o(1)}\right)\right)}$, where $\norm{A}^2$ is equivalent to a linear dependency in $T$, because the product of $\norm{A}^2$ and $T$ is a constant. We will give more details on these statements later in the proof and provide numerical evidence for Proposition~\ref{prop:Linear_scaling} in Figure~\ref{fig:Short_Range_Comparison}.

Furthermore, as already noted in Ref.~\cite{Ding2023-rv}, our objective might not necessarily be to precisely replicate the Lindblad dynamics but rather to prepare the ground state of $H$. Note that if we choose $L$, such that ${L\ket{\Psi_g}=0}$, we further have ${\mathcal{L}_L [\ket{\Psi_g}\bra{\Psi_g}] = 0}$ and hence the ground state is a fixed point of the LME~\ref{eqn:ground_state_lindblad_equation}:
\begin{equation}
\label{eqn:steady_state_argument}
    e^{\mathcal{L}_L t} [\ket{\Psi_g}\bra{\Psi_g}] = \ket{\Psi_g}\bra{\Psi_g}, \; \forall t\geq 0.
\end{equation}
Thus, we can tailor the algorithm to ground-state preparation by recognizing that the steady-state argument from Eq.~\ref{eqn:steady_state_argument} remains valid for the approximation in Eq.~\ref{eqn:sigma_lindblad_approximation}:
\begin{equation}
\label{eqn:discrete_time_convergence_holds}
\begin{aligned}
    &\tilde{L} \left(\ket{0}_a \bra{0}_a \otimes \rho_g\right) = \begin{pmatrix}
        0 \\ L \rho_g
    \end{pmatrix} = \begin{pmatrix}
        0 \\ 0
    \end{pmatrix} \\
    \Longleftrightarrow \; &\sigma(t)[\rho_g]=\Tr_a \left( e^{-i\tilde{L}\sqrt{t}} [\ket{0}_a \bra{0}_a \otimes \rho_g] e^{i \tilde{L} \sqrt{t}}\right) = \rho_g.
\end{aligned}
\end{equation}
We will show in Appendix~\ref{sct:Discrete-Time Dynamics} that this allows us to choose a larger time step $\tau$ and reduce the overall Hamiltonian simulation time, which is on the order of the total cost of implementation, to:
\begin{equation}
    T_H = \tilde{\Theta}\left(T^{1+o(1)} \norm{A}^{2+o(1)} \epsilon^{-o(1)}\right).
\end{equation}
To distinguish the new ansatz from the continuous-time approach that follows the LME, we will refer to the former as \emph{discrete-time} ground-state preparation algorithm. In summary, instead of following the LME~\ref{eqn:ground_state_lindblad_equation}, we now use the approximation scheme in Eq.~\ref{eqn:sigma_lindblad_approximation} and simply choose a larger time step $\tau$. This approach will still prepare the ground state if Eq.~\ref{eqn:steady_state_argument} holds.

We want to emphasize that alternative algorithms~\cite{ding2024simulating, cleve_et_al:LIPIcs.ICALP.2017.17, li2023simulating} result in an equivalent complexity with the advantage of not having to fulfill the steady state argument at all times. Nevertheless, they rely on cumbersome block encodings and precise control circuits, whereas in our case only a single ancilla qubit with simpler time evolutions is required. Further, note that the discrete-time mixing time is not guaranteed to possess an equivalent scaling as the continuous-time Lindblad dynamics in Eq.~\ref{eqn:ground_state_lindblad_equation}. Therefore, our refinement of the implementation cost from Proposition~\ref{prop:Linear_scaling} might prove beneficial. We will analyze the differences between the discrete-time algorithm and the Lindblad dynamics in Section~\ref{sct:Results for Grover's Search} for the Grover problem.

\section{Methods for Proving Convergence and Estimating the Mixing Time}
\label{sct:Methods for Proving Convergence and Estimating the Mixing Time}

In this section, we will introduce the methods employed in Section~\ref{sct:Results for Grover's Search} to prove the convergence of the LME to the ground state. Furthermore, we will explain how we determine the convergence speed, dividing the discussion into continuous-time dynamics from Eq.~\ref{eqn:original_lindblad_equation} and discrete-time dynamics, following Eq.~\ref{eqn:sigma_lindblad_approximation}.

\subsection{Continuous-Time Dynamics}
\label{sct:Continuous-Time Dynamics}

To determine the steady states of the Lindblad dynamics, we will either solve the LME~\ref{eqn:original_lindblad_equation} analytically, when possible, or use the following Lemma:
\begin{lemma}
\label{lma:Convergence}

A pure state $\ket{\Psi_i}$ is a steady state of the LME iff it is in the simultaneous kernel of all Lindblad jump operators, i.e., ${L_j \ket{\Psi_i} = 0, \; \forall j}$, and at the same time an eigenstate of the Hamiltonian. If there is no subspace $S$ which is left invariant under the Lindblad jump operators, i.e., ${L_j \ket{\Phi} \in S, \; \forall \ket{\Phi} \in S, \; \forall j}$, then the states $\ket{\Psi_i}$ are the only steady states.
\end{lemma}
\noindent A detailed proof can be found in Ref.~\cite{Kraus2008}. We state three distinct approaches that fulfill these conditions for Grover's search in Table~\ref{table:results}. The proofs of convergence to the ground state are provided in Section~\ref{sct:Continuous-Time Random Walk} and Appendix~\ref{sct:Quantum Continuous-Time Short-Range Convergence}.

The rate of convergence can be examined by analyzing the spectrum of $\mathcal{L}$. It can be shown~\cite{Rivas2012} that all nonzero eigenvalues of $\mathcal{L}$ have a negative real part. The real part determines -- in analogy to classical dynamical systems -- the rates of decay to the steady state, starting from different initial states~\cite{Minganti2018}. Consequently, we are interested in the eigenvalue $\alpha^*$, which is non-zero and has the largest real part, as thus defines the relevant relaxation time scale of the LME.

To determine $\alpha^*$ analytically, we rewrite the LME while reducing the Fock-Liouville space to a low-dimensional effective space, where the Lindbladian can be easily diagonalized. We do this by exploiting the symmetry of the underlying problem Hamiltonian. All calculations will be performed in the energy eigenbasis of $H$, where ${\rho_{ab} := \bra{a} \rho \ket{b}}$. Rewriting the LME~\ref{eqn:ground_state_lindblad_equation} results in:
\begin{widetext}
\begin{equation}
\begin{aligned}
\label{eqn:lindblad_equation_in_indices}
    \frac{\textrm{d}\rho_{ab}}{\textrm{d}t} &= \bra{a} \mathcal{L} [\rho] \ket{b} \\
&= \sum_{ij} \left( \hat{\gamma}_{ai} \hat{\gamma}_{bj} A_{ai} A_{bj} \rho_{ij} - \frac{1}{2} \left(\hat{\gamma}_{ji} \hat{\gamma}_{ja} A_{ji} A_{ja} \rho_{ib} + \hat{\gamma}_{ij} \hat{\gamma}_{ib} A_{ij} A_{ib} \rho_{aj} \right) \right).
\end{aligned}
\end{equation}
\end{widetext}
Here, we introduce the notations ${\hat{\gamma}_{ij}:=\hat{\gamma}(\lambda_i - \lambda_j)}$ and ${A_{ij}:=\bra{i}A\ket{j}}$. Finally, to gain a better understanding of the influence of coherence terms on the rate of convergence, we will further analyze the Lindblad equation in its diagonal form. For a non-degenerate spectrum, the diagonal elements of the density matrix $p_i := \rho_{ii}$ in Eq.~\ref{eqn:lindblad_equation_in_indices} decouple from the off-diagonal coherence terms \blue{(for details, see Ref.~\cite[Chapter~3]{Breuer2007})}:
\begin{equation}
    \label{eqn:lindblad_equation_diagonal}
    \frac{\textrm{d}p_i}{\textrm{d}t}=\sum_j \hat{\gamma}_{ij}^2  A_{ij}^2  p_{j} - \sum_j \hat{\gamma}_{ji}^2  A_{ji}^2  p_{i}.
\end{equation}
Eq.~\ref{eqn:lindblad_equation_diagonal} provides an approximation in cases, where the coherences decay quickly due to dephasing noise or when the degeneracies are broken by perturbations~\cite{Vogl_2010}. For example, room-temperature exciton transport properties have been investigated using this approach in Ref.~\cite{Davidson2020, Davidson2022}, where dephasing happens over a much shorter time scale than thermal relaxation. We will refer to it as the diagonal Lindblad Master equation (DLME). By performing a re-normalization step ${A_{ij}^2 \rightarrow \bar{A}_{ij}}$ and imposing ${\sum_i \bar{A}_{ij} = 0, \; \forall j}$, to ensure the conservation of probability, we arrive at:
\begin{equation}
    \label{eqn:pauli_master_equation}
    \frac{\textrm{d} p_i}{\textrm{d}t} = \sum_j \bar{A}_{ij}  p_j = \sum_{j \neq i} \left(\bar{A}_{ij}  p_j - \bar{A}_{ji} p_i\right).
\end{equation}
Eq.~\ref{eqn:pauli_master_equation} is known as Pauli master equation (PME) or rate equation. \blue{Note that although the DLME and PME yield the same dynamics, they require different cost calculation methods due to the re-normalization step. Additionally, since the PME enforces probability conservation, the filter function $\gamma$ can be omitted. The circuit under discussion follows the DLME, a physical equation that explains phenomena such as Dicke superradiance~\cite{Vogl_2010}, whereas the PME provides a framework for modeling a Markovian continuous-time CRW~\cite{Wong2022}. In Section~\ref{sct:Continuous-Time Random Walk}, we analyze the differences between these two equations, focusing on the quadratic speedup observed in Ref.~\cite{Vogl_2010} when transitioning from Eq.~\ref{eqn:lindblad_equation_in_indices} to Eq.~\ref{eqn:lindblad_equation_diagonal} for the Grover problem. We clarify that this speedup effectively reduces the implementation complexity from $\Theta(N^2)$ to $\Theta(N)$.} Furthermore, we will discuss the conditions under which the LME behaves (dis)similarly to the PME. Noteably, we will demonstrate that Eq.~\ref{eqn:lindblad_equation_in_indices} reduces to Eq.~\ref{eqn:pauli_master_equation} for the ansatz in Section~\ref{sct:Convergence under the Eigenstate Thermalization Hypothesis}.

\subsection{Discrete-Time Dynamics}

The discrete-time dynamics are no longer governed by the LME, but instead by Eq.~\ref{eqn:sigma_lindblad_approximation}. However, we can infer from Eq.~\ref{eqn:discrete_time_convergence_holds} that the steady state argument also holds for $\tilde{L}$ and hence we can employ the same ansätze as in the continuous-time dynamics. In the following, we will differentiate between the Multi-Trace (i.e., tracing out the ancilla many times) and the Single-Trace setting (i.e., performing a unitary evolution followed by a single trace-out operation at the end).

To estimate the rate of convergence, we need to find an analytical expression for the \mbox{$m$-th} step of the iterative Eq.~\ref{eqn:sigma_lindblad_approximation} that is dependent only on the initial state. This expression must then be solved for the total number $\mathcal{N}$ of iterations required to prepare the ground state and an appropriate time step $\sqrt{\tau}$ must be chosen. Hence, we define the discrete-time mixing time as the overall simulation time $T=\mathcal{N} \cdot \sqrt{\tau}$, after which the ground-state overlap ${\mu_g := \bra{\Psi_g}\sigma_{\mathcal{N}}(\tau)\ket{\Psi_g}}$ is larger than $1-\epsilon'$, for some desired accuracy $\epsilon'$ with $0\leq\epsilon'\leq1$. Note that the time step $\sqrt{\tau}$ is chosen in accordance with the convention adopted in Eq.~\ref{eqn:sigma_lindblad_approximation}.

Finally, we will compare the results to a Markovian discrete-time CRW, which can be modelled by~\cite{Wong2022}:
\begin{equation}
\label{eqn:classical_discrete_time_markov_chain}
    \vec{\mu}_t=\vec{\mu}_0 P^t,
\end{equation}
where $P$ is referred to as \emph{transition matrix} and must be chosen to be stochastic, that is, its entries are all non-negative and fulfill $\sum_j P_{ij}=1, \; \forall i$. $\vec{\mu}_t$ denotes the probability vector at step $t$.

\section{Results for Grover's Search}
\label{sct:Results for Grover's Search}

The system Hamiltonian we are examining may be seen as an unstructured search space, with exactly one marked element, which is the unique ground state $\ket{g}$, and exponentially many degenerate excited states:
\begin{equation}
\label{eqn:grover_hamiltonian}
    H = - \Delta  \ket{g}\bra{g}.
\end{equation}
For simplicity, we set ${\Delta=1}$. Note that the position of the marked element is unknown initially, and the objective is to identify it.

The algorithm detailed in Section~\ref{sct:Single-Ancilla Ground-State Preparation via Lindbladians} does not inherently ensure ergodicity and quantum detailed balance for such a system. However, we can achieve this by a suitable choice of coupling operator $A$ and filter function $\hat{\gamma}$. We refer to the different choices of $A$ and $\hat{\gamma}$ as transition or coupling regimes. Figure~\ref{fig:transition_regimes} offers a visualization of the different transition regimes. These are introduced in more detail in the following Subsection~\ref{sct:Continuous-Time Random Walk}.

An overview of the results that are derived in this section is provided in Table~\ref{table:results}. In the continuous-time long-range approach, where all excited states are connected to the ground state via the coupling operator $A$, we observe a quadratic speedup in quantum dynamics. This speedup, however, is only evident over the DLME and not the classical PME. When we employ the PME, we find that the overall complexity remains linear in $N$. The reason is that the quantum continuous-time algorithm is specifically designed to simulate the LME and not ground-state preparation, which is our actual objective. Nevertheless, we can recover the Grover speedup for the discrete-time dynamics in the Single-Trace setting, which is equivalent to a long unitary evolution and a single trace-out operation in the end. All other approaches result in a mixing time of $\Theta(N)$ (in case ${\norm{A}=1}$, which makes the mixing time equivalent to the overall complexity).

\begin{figure}
  \centering
  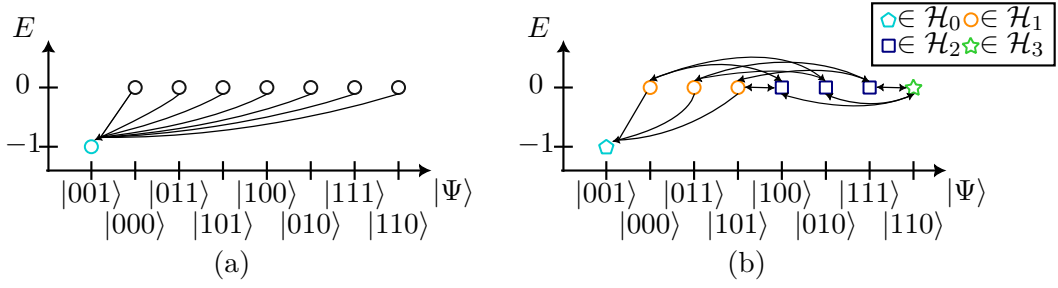
  \caption{\justifying A sketch of the (a) long-range and (b) short-range transition regimes for ${n=3}$ and ${\ket{g}=\ket{001}}$. The $\mathcal{H}_i$ are defined in Section~\ref{sct:Single Bit-Flip Operators}.}
  \label{fig:transition_regimes}
\end{figure}

\begin{table*}
\centering
\begin{ruledtabular}
\begin{tabular}{c c c c c c c}
 Ansatz & Transitions & $\norm{A}$ &  $\hat{\gamma}(0)$ & $T_{L}$ & $T_{D}$ & $T_P$ \\ [0.5ex] 
 \hline\hline
 Continuous-time & \multicolumn{6}{c}{} \\
 \hline
 ETH (\ref{sct:Convergence under the Eigenstate Thermalization Hypothesis}) & random & $\sqrt{N}$ & 0 & $\Theta(N)$ & $\Theta(N)$ & $\Theta(N)$ \\ 
 Single Projector (\ref{sct:A Single Projector}) & long-range & $N$ & 0 & $\Theta(N)$ & $\Theta(N^2)$ & $\Theta(N)$ \\
 \blue{Single Bitflip (\ref{sct:Single Bit-Flip Operators})} & \blue{short-range} & $n$ & 1 & \blue{$\sim\Theta(N)$} & \blue{$\sim\Theta(N)$} & \blue{$\sim\Theta(N)$} \\ [1ex] 
 \hline \hline
 Discrete-time & \multicolumn{6}{c}{} \\
 \hline
 Multi-Trace (\ref{sct:Multi-Trace}) & long-range & $N$ & 0 & $\Theta(N)$ & - & $\Theta(N)$ \\
 Single-Trace (\ref{sct:Single-Trace}) & long-range & $N$ & 0 & $\Theta(\sqrt{N})$ & - & - \\
\end{tabular}
\end{ruledtabular}
\caption{\justifying A summary of the results derived in this work. The continuous-time dynamics are categorized into random, long-range, and short-range transitions, defined by different choices of $A$ and $\hat{\gamma}$, which has to be chosen appropriately at $\omega=0$ to ensure convergence. For the discrete-time analysis, we focus on the long-range ansatz and distinguish between applying a trace-out operation many times (Multi-Trace) or only once at the end (Single-Trace). $T_{\framebox(2,2){}}$ represents the mixing time (with $\norm{A}=1$, making $T_{\framebox(2,2){}}$ equivalent to the overall complexity), and indices $\,\framebox(6,6){} \in \{L, D, P\}$ indicate that the dynamics follow the LME~\ref{eqn:lindblad_equation_in_indices} (or Eq.~\ref{eqn:sigma_lindblad_approximation} for discrete-time), DLME~\ref{eqn:lindblad_equation_diagonal}, or PME~\ref{eqn:pauli_master_equation} (or Eq.~\ref{eqn:classical_discrete_time_markov_chain}). In the continuous-time long-range approach, we observe a quadratic speedup for the LME over the DLME, but not the classical PME. \blue{Note that the scaling of the short-range ansatz has not been analytically proven but has been validated through numerical simulations with up to 48 qubits.} The discrete-time algorithm recovers a Grover speedup in the Single-Trace setting.}
\label{table:results}
\end{table*}

\subsection{Continuous-Time Random Walk}
\label{sct:Continuous-Time Random Walk}

In this subsection, we analyze the continuous-time approaches proposed in Table~\ref{table:results}. For each ansatz, we solve the LME, deduce the mixing time, and compare the dynamics to those derived from the classical PME. We demonstrate the equivalence between the LME~\ref{eqn:lindblad_equation_in_indices} and the PME~\ref{eqn:pauli_master_equation} in the Eigenstate Thermalization Hypothesis (ETH) approach and highlight the disparities in the long- and short-range coupling regimes. Furthermore, we specify the conditions on the initial state $\rho(0)$, for which the convergence to the ground state is assured.

In the following, $T_{\framebox(2,2){}}$ stands for the mixing time and the indices $\,\framebox(6,6){} \in \{L, D, P\}$ indicate that the dynamics follow the LME, DLME, or PME, respectively. In all calculations performed in Section~\ref{sct:Continuous-Time Random Walk}, we fix ${\norm{A}=1}$ to transfer all complexity to the mixing time $T$. Although Proposition~\ref{prop:Linear_scaling} states that the optimal scaling is sometimes recovered when shifting all complexity into $\norm{A}$, we believe that it is more intuitive to keep the mixing time as the computational bottleneck for the presentation of the results.

\subsubsection{Convergence under the Eigenstate Thermalization Hypothesis}
\label{sct:Convergence under the Eigenstate Thermalization Hypothesis}

The authors in Ref.~\cite{Ding2023-rv} propose an approach that leads to ergodicity in the expectation sense for arbitrary Hamiltonians and provide a proof that indicates that the mixing time of the LME might be polynomial in the system size $n$ under certain conditions. Their approach requires $A$ to be chosen independently at random from a normal distribution on the set of Hermitian matrices such that the entries fulfill ${\mathbb{E}(A_{ij})=0, \; \forall t \geq 0}$. In practice, this implies that the coupling matrix needs to be periodically chosen anew, typically after performing a certain number of iterations. We further need to assume that ${\sigma_{ij}:=\mathbb{E}(|A_{ij}|^2) > 0}$ to generate transitions. This ansatz can also be seen as a version of the Eigenstate Thermalization Hypothesis (ETH)~\cite{Srednicki1999, DAlessio2016}, which requires additional assumptions on the variance of $A_{ij}$.

For simplicity, we assume a sharp filter function with ${\hat{\gamma}(\omega) = 0}$, for ${\omega \geq 0}$ and ${\hat{\gamma}(\omega) = 1}$ otherwise, which is illustrated in Figure~\ref{fig:gamma_w_and_s_a}. Thus, we only allow transitions to lower-energy eigenstates. In reality, the assumptions given in Eq.~\ref{eqn:conditions_on_gamma} have to be fulfilled, to allow for an efficient implementation. A filter function that approximates the step function and is suitable for real application is given in Eq.~\ref{eqn:gamma_of_omega}. The effects of an imperfect $\hat{\gamma}$ are analyzed in Appendix~\ref{sct:Influence of the Hamiltonian}.

Before we can make statements about the rate of convergence, we need to restrict the coupling strength $\sigma_{ij}$. For a Hermitian matrix $A$, whose entries are drawn independently at random from a complex normal distribution with unit variance and mean zero, it can be shown that~\cite{terrencetao}
\begin{equation*}
\mathbb{E}\left(\norm{A}\right)=\mathcal{O}\left(K \sqrt{N}\right),
\end{equation*}
for some $K>0$, if the entries of $A$ fulfill the second and fourth moment bounds:
\begin{align}
\label{eqn:conditions_on_A_ETH}
\begin{split}
    \underset{i}{\sup}\, \mathbb{E}\left(A_{ij}^2\right) &\leq K^2 N, \\
    \underset{j}{\sup}\, \mathbb{E}\left(A_{ij}^2\right) &\leq K^2 N, \\
    \sum_{i,j} \mathbb{E}\left(A_{ij}^4\right) &\leq K^4 N^2, \; \forall i,j.
\end{split}
\end{align}
Therefore, we choose the coupling operator ${A \rightarrow \eta A}$, with ${\eta = \frac{1}{\sqrt{N}}}$, such that we can expect \blue{${\mathbb{E}\left(\norm{A}\right)=\mathcal{O}(1)}$} and thus shift all complexity into the mixing time $T_L$. Further, let $A$ fulfill Eq.~\ref{eqn:conditions_on_A_ETH}.
Now, define ${z_{g} := \rho_{gg}}$ and ${z_e^D := \sum_{a \neq g} \rho_{aa}}$ as the populations of the ground and excited states, respectively. Rewriting Eq.~\ref{eqn:lindblad_equation_in_indices} with those variables results in:
\begin{widetext}
\begin{equation}
\label{eqn:ETH_ODEs}
\begin{aligned}
    \mathbb{E}(\dot{z}_g) = \sum_{ij} \mathbb{E}(\hat{\gamma}_{gi} \hat{\gamma}_{gj} \underbrace{A_{gi} A_{gj}}_{=\sigma_{gi}\delta_{ij}} \rho_{ij}&-\frac{1}{2}(\hat{\gamma}_{ji} \underbrace{\hat{\gamma}_{jg}}_{=0} A_{ji} A_{jg} \rho_{ig} + \hat{\gamma}_{ij} \underbrace{\hat{\gamma}_{ig}}_{=0} A_{ij} A_{ig} \rho_{gj})) \\ =\eta^2  \mathbb{E}\left(\sum_{a \neq g} \rho_{aa}\right) &= \eta^2 \mathbb{E}\left(z_e^D\right) \\
    \mathbb{E}(\dot{z}_e^D)&=-\eta^2  \mathbb{E}(z_e^D).
\end{aligned}
\end{equation}
\end{widetext}
Here, we introduced the Kronecker delta $\delta_{ij}$ and used the fact that ${\mathbb{E}(A_{ij}  A_{mn}) = 0}$, if ${(i,j) \neq (m,n)}$. This calculation allows to estimate the mixing time $T_L$, by solving the system of ordinary differential equations (ODEs)~\ref{eqn:ETH_ODEs}:
\begin{equation*}
\begin{aligned}
    \mathbb{E}(z_g(t))&=z_e^D(0)(1-\underbrace{e^{-\eta^2 t}}_{=e^{-\frac{1}{T_L} t}}) + z_g(0) \\
    \mathbb{E}(z_e^D(t))&=z_e^D(0) e^{-\eta^2 t}.
\end{aligned}
\end{equation*}
We see that the approach converges to the ground state for any initial state, since ${\lim_{t \rightarrow \infty} \mathbb{E}(z_g(t)) = z_g(0) + z_e^D(0) = 1}$.
Recall from Section~\ref{sct:Continuous-Time Dynamics} that we defined the mixing time $T_L$ via the rates of decay to the steady state. Hence, we yield ${T_{L}=\frac{1}{\eta^2}=N}$ on average. As mentioned at the beginning of this subsection, Theorem 21 in Ref.~\cite{Ding2023-rv} indicates a scaling of ${T_L=\tilde{\mathcal{O}}(\mathrm{poly}(n))}$ for the ETH-type approach. However, the authors assume transition rates that are not exponentially small in the system size, specifically $\sigma_{ij}=\mathcal{O}\left(1/\mathrm{poly}(n)\right)$, which leads to an exponentially large $\norm{A}$ for our setting. We will prove in Proposition~\ref{prop:Linear_scaling}, that the complexity can be shifted between $\norm{A}^2$ and $T_L$, but the overall complexity in terms of the total Hamiltonian simulation time $T_H=\tilde{\Theta}\left(T_L \norm{A}^2\right)$ remains constant. Consequently, this implies that the ETH-type approach performs comparably to classical exhaustive database search.

Indeed, we can show that the LME reduces to the coherence-free DLME on average:
\begin{equation}
    \label{eqn:ETH_reduction_to_pauli}
    \frac{\textrm{d}\mathbb{E}(\rho_{ab})}{\textrm{d}t} = \mathbb{E}(\mathcal{L}[\rho_{ab}]) = \sum_{i\neq g} \delta_{ag}  \sigma_{ai}  \rho_{ii} - \frac{1}{2}  (\sigma_{ga} + \sigma_{gb})  \rho_{ab}.
\end{equation}
We observe that this approach yields the same equations of motion as the CRW governed by the PME~\ref{eqn:pauli_master_equation}. This result can easily be extended to more general Hamiltonians.

\subsubsection{A Single Projector}
\label{sct:A Single Projector}

Next, we choose ${A=\eta J_N}$, where $J_N$ denotes the all-one matrix of dimensions $N \times N$, i.e., $A_{ij}=1$. When norming $A$ by setting $\eta=1/N$, this ansatz is equivalent to ${A=\ket{s}\bra{s}}$, where ${\ket{s} = \frac{1}{\sqrt{N}} \sum_{i} \ket{i}}$ is the uniform superposition of all computational basis states. In contrast to Section~\ref{sct:Convergence under the Eigenstate Thermalization Hypothesis}, the transitions in this section are not chosen randomly, but remain constant throughout the protocol. The long-range approach has been analyzed in an earlier work and a quadratic quantum speedup from $\mathcal{O}\left(N^2\right)$ for the DLME to $\mathcal{O}(N)$ for the LME has been found~\cite{Vogl_2010}. The authors note that stringent restrictions on the operator norm of the coupling operator could prevent achieving a Grover speedup. We address this open question by mapping the results to the implementation presented in Section~\ref{sct:Single-Ancilla Ground-State Preparation via Lindbladians}, thereby properly evaluating the complexity.

Following an identical procedure as employed in the previous Subsection~\ref{sct:Convergence under the Eigenstate Thermalization Hypothesis} to derive the rate of convergence, we define the variables ${z_g:=\rho_{gg}}$ and ${z_{e}:=\sum_{\substack{a \neq g \\ b \neq g}} \rho_{ab}}$. Note, that $z_e$ now includes off-diagonal entries. By rewriting Eq.~\ref{eqn:lindblad_equation_in_indices} with these variables, we get to the following system of ODEs:
\begin{equation}
\begin{aligned}
\label{eqn:single_projector_ODEs}
    \dot{z}_g &= \dot{\rho}_{gg} = \eta^2  \sum_{\substack{a \neq g \\ b \neq g}} \rho_{ab} = \frac{1}{N^2}  z_e \\
    \dot{z}_{e} &= \sum_{\substack{a \neq g \\ b \neq g}} \dot{\rho}_{ab} = -(N-1)  \eta^2 z_{e} = - \frac{N-1}{N^2}  z_{e}.
\end{aligned}
\end{equation}
After solving them, we find that this setup results in a mixing time of ${T_L = \frac{N^2}{N-1} = \Theta(N)}$:
\begin{equation*}
\begin{aligned}
    z_g(t)&=z_g(0) + \frac{z_e(0)}{N-1} (1-\underbrace{e^{-(N-1)t/N^2}}_{=e^{-\frac{1}{T_L} t}}) \\
    z_e(t)&=z_e(0) e^{-(N-1)t/N^2}.
\end{aligned}
\end{equation*}
We further see that convergence to the ground state is guaranteed if we have ${\lim_{t\rightarrow \infty} z_g(t) = z_g(0) + \frac{z_e(0)}{N-1} \overset{!}{=} 1}$, which is for instance fulfilled for $\rho(0)=\ket{s}\bra{s}$.

When decoupling the probabilities from the coherences in Eq.~\ref{eqn:single_projector_ODEs} by lifting the degeneracies of $H$, we see that the convergence time of the DLME scales as ${T_D=\Theta(N^2)}$. After a re-normalization step of $A$, that has been described in Section~\ref{sct:Continuous-Time Dynamics}, we arrive at the ODEs for the PME:
\begin{equation}
\begin{aligned}
\label{eqn:single_projector_pauli_ODEs}
    \mu_g &= \frac{1}{N} \sum_{a\neq g} \mu_a \\
    \sum_{a\neq g} \mu_a &= - \frac{1}{N}  \sum_{a\neq g} \mu_a.
\end{aligned}
\end{equation}
From this we can infer that the CRW converges after the time ${T_P=\Theta(N)}$. Summarizing, we observe a quadratic quantum speedup for the LME in comparison the DLME, but not to a CRW (PME). Hence, we are not able to recover the Grover speedup in this setting.

Note that Eq.~\ref{eqn:single_projector_pauli_ODEs} can be derived directly from the PME without the reduction of the LME to the DMLE and the re-normalization of the DLME. The approach is equivalent to a spatial search on an all-to-all network~\cite{Wong2022}. Intuitively, a walker hops from vertex to vertex, querying an oracle along the way, until the marked state has been located. Mathematically, this can be modeled by choosing the coupling matrix as the Laplacian ${A=\frac{M-D}{\norm{M-D}}}$ of a fully connected graph with an absorbing vertex at $g$. $M$ denotes the adjacency matrix, and $D$ is the degree matrix. Employing the spatial-search ansatz in the LME results in the same ODEs as in Eq.~\ref{eqn:single_projector_ODEs}. To see this, one has to redefine the filter function to additionally allow transitions between degenerate eigenstates: ${\hat{\gamma}(\omega)=0,\;\forall \omega>0}$ and one otherwise (i.e., $\hat{\gamma}(0)=1$). Thus, the coupling between degenerate eigenstates and the diagonal entries of $A$ do not influence the dissipative Lindblad dynamics in the long-range approach.

In contrast to the ETH-type ansatz, we find that for the long-range approach, it is not possible to decouple the coherences from the quantum LME without lifting the degeneracies. After defining ${z_e^D := \sum_{a\neq g} \rho_{aa}}$ and ${z_e^O:=\sum_{i \neq g} \sum_{\substack{j \neq g \\ j \neq i}} \rho_{ij}}$, we get:
\begin{equation}
\label{eqn:single_projector_reduction_to_pauli}
    \dot{z}_e^D = \sum_{a\neq g} \dot{\rho}_{aa} = -\frac{1}{2} \sum_{\substack{a \neq g \\ i \neq g}} A_{gi}  A_{ga}  (\rho_{ia} + \rho_{ai}) = -\eta^2 (z_e^D + z_e^O).
\end{equation}
Hence, the quantum dynamics of the LME are different to those of the PME. Nevertheless, as it can be seen from Eq.~\ref{eqn:single_projector_ODEs}, the joint ODE for the diagonal and coherence terms $z_e$ exhibit an equivalent form as the PME in Eq.~\ref{eqn:single_projector_pauli_ODEs}.

Finally, we want to highlight the difference between the unitary evolution, which results in a quadratic speedup, and the dissipative approach. In the case of the former, by choosing ${H=-\ket{g} \bra{g} - \frac{1}{N} A}$, where $A$ denotes the previously defined graph Laplacian, the dynamics reach a ground-state overlap of ${\mu_g=1}$ at ${t = \frac{\pi}{2} \sqrt{N}}$. This is quadratically faster than the classical analog~\cite{Wong2022}. Note that the propagation does not converge to a steady state, but performs a periodic oscillation. When the jump operator $L$ is chosen to be Hermitian in the purely dissipative framework -- similar to the Hamiltonian in unitary dynamics that results in a quadratic speedup -- the maximally mixed state is always the steady state of the evolution~\cite{Kraus2008}.

\subsubsection{Single Bit-Flip Operators}
\label{sct:Single Bit-Flip Operators}

Finally, we investigate the short-range transition regime with the ansatz ${A=\eta  \sum_i X_i}$, where $X_i$ is the Pauli-$X$ operator, applied to the \mbox{$i$-th} qubit. See Figure~\ref{fig:transition_regimes} for an intuitive illustration of the coupling regime. Since $A$ is Hermitian, we have ${\norm{A} = |\lambda_{\mathrm{max}}(A)| = \sum_i |\lambda_{\mathrm{max}}(X_{i})| = n=\log_2 N}$, where $\lambda_{\mathrm{max}}(A)$ denotes the largest eigenvalue of the matrix $A$. Thus, in order to achieve ${\norm{A} = 1}$, we set ${\eta = 1/n}$. Note that in comparison to the previous sections, $\eta$ is not exponentially small.

To simplify the calculations, we introduce the following notions. When dividing the Hilbert space into subspaces of states ${\mathcal{H} = \{\mathcal{H}_0, ..., \mathcal{H}_{\alpha}, ..., \mathcal{H}_{n}\}}$, with Hamming distance $\alpha$ to the solution, each subspace is spanned by ${\binom{n}{\alpha}}$ basis states. Further, for every state in ${\mathcal{H}_{\alpha}}$, there are ${(n-\alpha)}$ single-bitflips that lead to ${\mathcal{H}_{\alpha+1}}$ and $\alpha$ single-bitflips that lead to ${\mathcal{H}_{\alpha-1}}$.

To estimate the rate of convergence, let us define ${z^{\alpha'}_{\alpha}:=\sum_{i \in \mathcal{H}_{\alpha'}} \sum_{j \in \mathcal{H}_{\alpha}} \rho_{ij}}$. After deriving the system of ODEs for these variables from the LME, we see that $z_0^0=\rho_{gg}$ is coupled to ${z^{\alpha+2k}_{\alpha}}$ \& ${z^{\alpha}_{\alpha+2k}}$ for ${k \in \mathbb{N}}$, ${0\leq\alpha+2k\leq n}$. The latter are coupled to ${z^{\alpha+4k}_{\alpha}}$ \& ${z^{\alpha}_{\alpha+4k}}$, respectively, which are further coupled to ${z^{\alpha+6k}_{\alpha}}$ \& ${z^{\alpha}_{\alpha+6k}}$ and so on. Hence, the system of ODEs coupled to $z_0^0$ consists of ${\frac{1}{2} (n^2+2n+2)}$ variables if $n$ is even and ${\frac{1}{2} (n+1)^2}$ variables, if $n$ is odd. The expression is very lengthy and thus, we state it in Appendix~\ref{sct:Quantum Continuous-Time Short-Range Dynamics} in Eq.~\ref{eqn:short_range_ODEs}.
\begin{figure}[t]
  \centering
  \subfloat[Short-range scaling]{\includegraphics{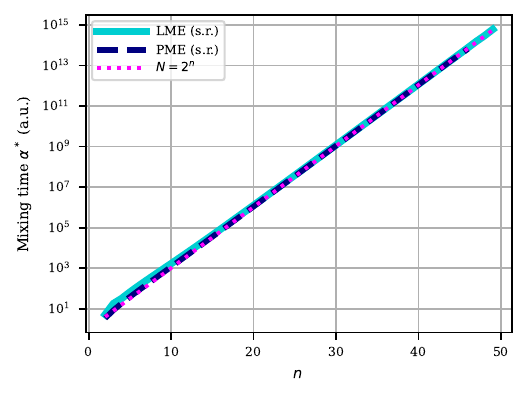}}\label{fig:short_range_scaling} \\
\subfloat[Comparison of dynamics]{\includegraphics{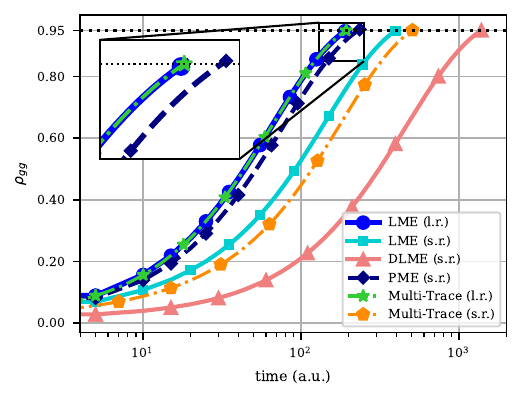}\label{fig:Short_Range_Comparison}}
  \caption{\justifying The scaling behavior of the mixing time $\alpha^*$ for short-range (s.r.) transitions (Section~\ref{sct:Single Bit-Flip Operators}) of the form ${A=\frac{1}{n} \sum_i X_{i}}$ are depicted in (a). They are calculated by taking the eigenvalue with the second-to-smallest real part of the coefficient matrices of Eqs.~\ref{eqn:short_range_pauli_ODEs} and~\ref{eqn:short_range_ODEs}. The mixing time of the LME is plotted alongside that of the classical PME. In (b), a comparison of the dynamics of the analyzed approaches is provided for a system size of ${n=6}$ and the approximate filter function in Eq.~\ref{eqn:gamma_of_omega}. The long-range (l.r.) approach was analyzed in Section~\ref{sct:A Single Projector} and the Multi-Trace dynamics are examined in Section~\ref{sct:Multi-Trace}. The simulation is implemented by directly calculating $e^{\mathcal{L}_L t}\rho(0)$, with $\rho(0)=\ket{s}\bra{s}$, and is interrupted at a ground-state overlap of $\rho_{gg}=0.95$.}
\end{figure}
The eigenvalue with the second largest real part of the reduced Fock-Liouville space and thus the convergence speed is evaluated numerically. The results can be found in Figure~\ref{fig:short_range_scaling} and indicate a scaling of $T_L = {\Theta(N)}$. For the single-bitflip coupling operator, we have $\norm{L}=\norm{A}$ for any $\eta$ and hence it is evident from Eq.~\ref{eqn:continuous_time_overall_complexity_lindblad} that the overall complexity in terms of the total Hamiltonian simulation time becomes $T_H=\tilde{\Theta}\left(N^2\right)$. Nevertheless, alternative implementations, such as the one presented in Ref.~\cite{ding2024simulating}, reduce the cost to a linear dependency in the mixing time (i.e., the overall complexity is given by ${\tilde{\mathcal{O}}\left(\norm{A}^2 T_L\right)=\tilde{\mathcal{O}}(N)}$), but are expected to necessitate fully fault-tolerant quantum computers.

Since we are not able to solve the ODEs analytically, we employ Lemma~\ref{lma:Convergence} to prove convergence to the ground state when starting in $\rho(0)=\ket{s}\bra{s}$. The proof is presented in Lemma~\ref{lma:short_range_convergence_proof}, which can be found in Appendix~\ref{sct:Quantum Continuous-Time Short-Range Convergence}.

An approach similar to the short-range coupling was analyzed in Ref.~\cite{Vogl_2010}, with the difference, that they shifted the energy levels of the Hamiltonian to obtain a Hamming ladder. The authors note that their ansatz leads to a logarithmic time complexity of $\mathcal{O}(\log N)$, which can also be achieved by the following classical greedy algorithm: For every bit in the configuration propose flipping it. Accept the bit flip only if the energy decreases. This way, the ground state is found after at most ${n=\log_2 N}$ bit flips for any initial configuration. Further, note that although the quantum Lindblad equations in Ref.~\cite{Vogl_2010} work for the Hamming ladder approach, they are not applicable in general. We derived the full equations and state them in Appendix~\ref{sct:Quantum Continuous-Time Short-Range Dynamics}, Eq.~\ref{eqn:short_range_ODEs}.

To derive the classical convergence time of the PME for the Grover Hamiltonian, we define the variables ${z_{\alpha}:=\sum_{i \in \mathcal{H}_{\alpha}} \mu_i}$. Then, the ODEs take the following form:
\begin{equation}
    \begin{aligned}
    \label{eqn:short_range_pauli_ODEs}
        \dot{z}_0 &= \eta  z_1 \\
        \dot{z}_1 &= \eta  (-n  z_1 + 2  z_2) \\
        \vdots \\
        \dot{z}_{\alpha} &= \eta \left((n-\alpha+1)  z_{\alpha-1} - n  z_{\alpha} + (\alpha+1)  z_{\alpha+1}\right) \\
        \vdots \\
        \dot{z}_n &= \eta  (z_{n-1} - n  z_n).
    \end{aligned}
\end{equation}
The scaling is again analyzed numerically by diagonalizing the coefficient matrix in Eq.~\ref{eqn:short_range_pauli_ODEs}. The results are provided in Figure \ref{fig:short_range_scaling} and strongly indicate a scaling of the mixing time of ${\Theta(N)}$ for $\eta=1/n$. Ref.~\cite{Allahverdyan2022} presented an approach yielding a complexity of $\mathcal{O}(\log_2 N)$ for the PME with long-range bounded transitions. This is achieved by introducing a shift in the energy spectrum of the Grover Hamiltonian, which is no larger than the spectral gap. Additionally, they argue that short-range transitions exhibit diffusive dynamics, resulting in convergence after a time of $\mathcal{O}\left(N^2\right)$. In contrast, our approach achieves a mixing time of $\Theta(N)$ even for the short-range regime.

Note that the quantum ODEs in Eq.~\ref{eqn:short_range_ODEs} and the classical ODEs of Eq.~\ref{eqn:short_range_pauli_ODEs} exhibit fundamentally different forms. Although Figure~\ref{fig:short_range_scaling} indicates that they approach each other asymptotically, we were not able to directly reduce the quantum equations to the classical ones. We hope that understanding this distinct behavior (and the one in the long-range ansatz in Eq.~\ref{eqn:single_projector_reduction_to_pauli}) will help identify problems that could benefit from dissipative approaches.

We further simulate the ground-state overlap numerically, starting in an equal superposition, up to $\mu_g=0.95$ in Figure~\ref{fig:Short_Range_Comparison} for most analyzed quantum and classical approaches with a system size of $n=6$ qubits. It is evident that the classical PME short-range dynamics deviate from the quantum LME short-range ansatz. Furthermore, note that the LME and discrete-time Multi-Trace dynamics behave equivalently for the long-range approach and deviate in the short-range regime, which provides numerical evidence supporting Proposition~\ref{prop:Linear_scaling}.

\subsection{Discrete-Time Random Walk}
\label{sct:Discrete-Time Random Walk}

We now analyze the more efficient discrete-time ground-state preparation algorithm introduced in Section~\ref{sct:Single-Ancilla Ground-State Preparation via Lindbladians} and presented in more detail in Appendix~\ref{sct:Discrete-Time Dynamics}. Here, we restrict our analysis to the long-range approach ${A=\ket{s}\bra{s}}$ already analyzed for the continuous-time dynamics in Section~\ref{sct:A Single Projector} and differentiate between tracing out multiple times (Multi-Trace) or only once at the end (Single-Trace).

\subsubsection{Multi-Trace}
\label{sct:Multi-Trace}

We prove in Appendix~\ref{sct:Discrete-Time Long-Range Approach_Appendix} that starting in ${\sigma_0 := \ket{s}\bra{s}}$, after $m$ steps of propagating via Eq.~\ref{eqn:sigma_lindblad_approximation}, we arrive at a density matrix of the form:
\begin{widetext}
\begin{equation}
\begin{aligned}
\label{eqn:single_projector_discrete_time_density_matrix}
    \sigma(m \tau) &= \Tr_a \left(e^{-i\tilde{L} \sqrt{\tau}} [\ket{0}_a\bra{0}_a \otimes \sigma_{m-1}(\tau)] e^{i\tilde{L} \sqrt{\tau}}\right) \\
    &=\ket{s}\bra{s} + \tilde{N} \left(\cos(\zeta)-1\right) \left(\sum_{k=0}^{m-1} \cos^k(\zeta)\right) \left(\ket{\tilde{s}}\bra{s} + \ket{s}\bra{\tilde{s}}\right) \\
    &\;\;\;\;+ \tilde{N}^2 \left(1-2 \cos^m(\zeta) + \cos^{2m}(\zeta)\right) \ket{\tilde{s}} \bra{\tilde{s}} + \tilde{N}^2 \left(1-\cos^{2m}(\zeta)\right) \ket{g} \bra{g}.
\end{aligned}
\end{equation}
\end{widetext}
Here, ${\ket{\tilde{s}} := \frac{1}{\sqrt{N-1}} \sum_{j \neq g} \ket{j}}$ is the uniform superposition state, excluding $\ket{g}$. Further, we defined ${\tilde{N} := \sqrt{\frac{N-1}{N}}}$ and ${\zeta := \eta \sqrt{N-1} \sqrt{\tau}}$. Hence, in step ${\mathcal{N}:=T/\tau}$, the system exhibits a ground-state overlap of:
\begin{equation}
\label{eqn:single_projector_discrete_time_ground_state_overlap}
\begin{aligned}
    \mu_{g}(\mathcal{N}, \tau) &= \bra{g} \sigma_{\mathcal{N}} \ket{g} \\
    &= \frac{1}{N} + \frac{N-1}{N} \left(1-\cos^{2\mathcal{N}}(\zeta)\right) \overset{!}{\geq} 1-\epsilon'.
\end{aligned}
\end{equation}
We estimate the complexity of the Multi-Trace approach by solving this inequality for $\mathcal{N}$ and multiply by $\sqrt{\tau}$ to yield a lower bound on the overall simulation time required for ground-state preparation:
\begin{equation}
\begin{aligned}
\label{eqn:ham_sim_time_discrete_quantum}
    &T=\mathcal{N} \sqrt{\tau} \gtrapprox - \frac{\log\left(\epsilon'^{-1}\right) \sqrt{\tau}}{2 \log\left(\cos\left(\eta \sqrt{N-1} \sqrt{\tau}\right)\right)} \\
    &= \log\left(\epsilon'^{-1}\right) \left(\frac{N^2}{(N-1)\sqrt{\tau}} - \frac{\sqrt{\tau}}{6} + \mathcal{O}\left(\frac{N-1}{N^2} \sqrt{\tau}^3\right)\right).
    \end{aligned}
\end{equation}
In the inequality, we approximate ${\frac{\epsilon'}{N-1} + \epsilon' \approx\epsilon'}$ and in the third step, we set $\eta=1/N$ to shift all complexity in the mixing time. We further use a Laurent series expansion of the term ${\log^{-1}\left(\cos\left(\frac{\sqrt{N-1}}{N}\sqrt{\tau}\right)\right)}$ at ${N=\infty}$. We can see that the discrete-time Lindblad dynamics scale as ${\tilde{\Theta}\left(N\right)}$ when setting $\sqrt{\tau}=1$, thus tracing out the ancilla qubit many times and fully shifting the complexity to $\mathcal{N}$.

Ref.~\cite{Wong2022} demonstrates that for the classical dynamics described by Eq.~\ref{eqn:classical_discrete_time_markov_chain}, the discrete-time CRW coupled by the graph Laplacian of an all-to-all connected graph with an absorbing vertex at $g$ results in:
\begin{equation}
\label{eqn:classical_discrete_time_ground_state_overlap}
    \mu_{g}(\mathcal{N}_C,1)=1-\frac{N-1}{N}\log\left(\frac{N-2}{N-1}\right)^{\mathcal{N}_C} \overset{!}{\geq} 1- \epsilon'.
\end{equation}
After solving for $\mathcal{N}_C$ and employing the Taylor series expansion, we see that ${{\mathcal{N}_C} = \Theta\left(N \log\left(\epsilon'^{-1}\right)\right)}$.

The quantum discrete-time long-range ground-state overlap from Eq.~\ref{eqn:single_projector_discrete_time_ground_state_overlap} is depicted in Figure \ref{fig:Short_Range_Comparison} for a system size of $n=6$ together with the continuous-time random walks. The former exhibits strong similarities to the continuous-time quantum approach. Therefore, the proposed early fault-tolerant discrete-time algorithm presented in Section~\ref{sct:Discrete-Time Dynamics} appears to not alter the mixing time significantly, which was to be expected from Proposition~\ref{prop:Linear_scaling}.

\subsubsection{Single-Trace}
\label{sct:Single-Trace}

Finally, we want to demonstrate that the quadratic Grover speedup can be recovered when the trace operation is conducted only once after an evolution time of $\tilde{\Theta}\left(\sqrt{N}\right)$. Here, we will additionally need to adapt the implementation procedure slightly to simulate the unitary evolution under $\tilde{L}$ for such a large time.

We can again make use of Eq.~\ref{eqn:ham_sim_time_discrete_quantum} to estimate the simulation time $T$. It is easy to see, that selecting $\mathcal{N}=\mathcal{O}(1)$ and $\sqrt{\tau}=\Theta\left(\sqrt{N}\right)$ minimizes $T$ to $\tilde{\Theta}\left(\sqrt{N}\right)$, which is quadratically faster than the classical dynamics in Eq.~\ref{eqn:classical_discrete_time_ground_state_overlap} and the continuous-time approaches in Section~\ref{sct:Continuous-Time Random Walk}.

To simulate $\tilde{L}$ for such a large time, we need to divide the unitary evolution into $r$ steps:
\begin{equation*}
    e^{i\tilde{L}\sqrt{\tau}} \approx \left(e^{i\tilde{L}\sqrt{\tau}/r}\right)^r.
\end{equation*}
The error of this approximation scheme is given by the sum of the error arising in the truncation \& discretization of the integral, and the error coming from the Trotter splitting. We will see later in Eq.~\ref{eqn:integral_truncation_and_discretization_error}, that we can neglect the former while the latter results in:
\begin{equation*}
\begin{aligned}
    \tnorm{\tilde{\chi}^r(\tau/r) - \tilde{\kappa}(\tau)} &= \mathcal{O}\left(\left(\mathcal{N}\sqrt{\tau} \norm{A}\right)^{p+2} r^{-(p+1)}\right) = \mathcal{O}(\epsilon),
\end{aligned}
\end{equation*}
where $\tilde{\kappa}(t)$ and $\tilde{\chi}(t)$ are defined by Eq.~\ref{eqn:discretized_integral_jump_operator_definition} and~\ref{eqn:final_approximated_density_matrix_chi} as the density matrices after approximating $\sigma(t)$ by truncation \& discretization of the integral and Trotterization, respectively. For a derivation of the first equation, we refer the reader to Appendix~\ref{sct:Trotter Splitting}. Note that here we additionally include the term $r^{-(p+1)}$ in the same manner as Ref.~\cite[Lemma 1]{Childs2021}. In the second equality, we choose ${r=\mathcal{O}\left((\mathcal{N} \sqrt{\tau} \norm{A})^{1+\frac{1}{p+1}} \epsilon^{-\frac{1}{p+1}}\right)}$ to achieve an accuracy of $\mathcal{O}(\epsilon)$ in our approximation scheme. The overall complexity in terms of the total Hamiltonian simulation time can hence be given as:
\begin{equation*}
\begin{aligned}
    T_H&=\Theta\left(\mathcal{N}\sqrt{\tau} \, r\right)=\tilde{\Theta}\left(\left(\norm{A}T_H\right)^{1+o(1)}\epsilon^{-o(1)}\right) \\
    &=\tilde{\Theta}\left(\sqrt{N}^{1+o(1)} \epsilon^{-o(1)}\right).
\end{aligned}
\end{equation*}
The dependency in $\epsilon$ can further be reduced to an additive logarithmic order via QSVT~\cite{Childs2018}. We see that the quadratic Grover speedup can be recovered, when using a single final trace-out operation. Eq.~\ref{eqn:ham_sim_time_discrete_quantum} further shows that each additional trace-out operation slows convergence, thereby illustrating the crucial role of coherence in Grover's algorithm. Note that the formula for the ground-state overlap $\mu_g$ in Eq.~\ref{eqn:single_projector_discrete_time_ground_state_overlap} for $m=1$ follows a similar structure as for an analog quantum random walk with $H=\ket{g}\bra{g} + \ket{s}\bra{s}$~\cite{PhysRevA.57.2403}:
\begin{equation*}
    \mu_g(1,T)=\frac{1}{N} + \frac{N-1}{N} \sin^2\left(\frac{\sqrt{N-1}}{N} T\right).
\end{equation*}
Therefore, we have deepened the understanding of the role of coherence in random walks by establishing a method to transition from continuous-time OQRWs to UQRWs. Further note that this leads to a different implementation of achieving a quadratic speedup through the time-evolution of a dilated jump operator $\tilde{L}$. This could open an exciting new research direction in the field of continuous-time UQRWs, which remains an active area of investigation~\cite{PhysRevLett.129.160502}. Note that the analysis and proof of convergence for the short-range ansatz (Section~\ref{sct:Single Bit-Flip Operators}) paves the way to extending this result. This can lead to a simpler implementation scheme utilizing a sum of single Pauli-$X$ gates as coupling operator.

\section{Conclusions and Future Research}
\label{Conclusions and Future Research}

This paper provides an intuitive presentation and a rigorous analysis of an early fault-tolerant quantum algorithm for ground-state preparation via Lindbladians. We demonstrated in Proposition~\ref{prop:Linear_scaling} that by employing a high-order Trotterization, we can improve the complexity of implementation for some instances to be linear in the mixing time and still follow the Lindblad dynamics. This can be valuable for simulating the LME or preparing ground-states via Lindbladians with limited quantum resources, which is of considerable interest in the community~\cite{ding2024simulating, cleve_et_al:LIPIcs.ICALP.2017.17, li2023simulating}. Then, we evaluate approaches that prepare ground states with the analyzed algorithm. Our findings provide strong evidence that in a purely dissipative continuous-time setting, a superlinear speedup over classical search is not possible. Thus, we demonstrate that the previously observed quadratic speedup~\cite{Vogl_2010} arises not in comparison to the CRW, but relative to the unnormalized DLME. This leads to well-known physical phenomena, such as Dicke superradiance.

We further reduce the quantum LME to the classical PME in the ansatz that is inspired by the ETH (Section~\ref{sct:Convergence under the Eigenstate Thermalization Hypothesis}) and show that the classical and quantum ODEs differ from each other in the long- and short-range regimes (Sections~\ref{sct:A Single Projector} and~\ref{sct:Single Bit-Flip Operators}, respectively). We aim to have thereby deepened the foundational insights into dissipative dynamics for ground-state preparation. Initially, analogous deviations from the classical dynamics were observed in the quantum adiabatic algorithm~\cite{farhi2000quantum}, which then proved useful for many problems when the evolution path is properly controlled~\cite{Avron2010}.

Nevertheless, we are able to recover a quadratic speedup in the unitary limit, meaning that we trace-out the ancilla qubit only once at the end of the evolution. On the one hand, this leads to a different way of implementing continuous-time UQWRs via the dilated jump operator $\tilde{L}$, and on the other hand reveals the critical role of coherence for the quadratic Grover speedup, a topic that has been explored extensively in previous studies~\cite{Sun2024, Rastegin2018_1, Rastegin2018_2} and establishes a transition from OQRWs to UQRWs.

A noteworthy finding for future investigations is that when dephasing is considered, there is a unique schedule which maximizes the fidelity with a target state in adiabatic evolution~\cite{Avron2010}. Ref.~\cite{King2024} analyzed this behavior with the goal of finding shortcuts to adiabaticity. Hence, future research could explore the combination of unitary and dissipative dynamics, incorporating potentially time-dependent Lindbladians as $\mathcal{L}(t):=\mathcal{L}_H(t)+\mathcal{L}_L(t)$.

\section*{Acknowledgements}

P.J.E and J.R.F were partially funded by the German BMWK project QCHALLenge (01MQ22008B and 01MQ22008D, respectively). The authors thank Lin Lin and Pedro Hack for insightful discussions.

\appendix

\section{Single-Ancilla Ground-State Preparation via Lindbladians}
\label{sct:Single-Ancilla Ground-State Preparation via Lindbladians_Appendix}

This section presents an overview of the algorithm introduced in Ref.~\cite{Ding2023-rv} to simulate open-system dynamics on an early fault-tolerant digital quantum computer. We aim to provide a concise summary, emphasizing the intuitive understanding behind each building block of the algorithm. Furthermore, we highlight its applicability to prepare ground states of Hamiltonians.

Remember from Section~\ref{sct:Single-Ancilla Ground-State Preparation via Lindbladians} that the goal of the algorithm is to simulate $\rho(t)$ in the purely-dissipative Lindblad dynamics with one jump operator $L$:
\begin{equation}
\label{eqn:ground_state_lindblad_equation_Appendix}
    \frac{\textrm{d}\rho}{\textrm{d}t} = \mathcal{L}_L[\rho] = L \rho L^\dag - \frac{1}{2} \{L^\dag L, \rho\}.
\end{equation}
Since Eq.~\ref{eqn:ground_state_lindblad_equation_Appendix} admits a formal solution ${\rho(t)=e^{\mathcal{L}t}[\rho(0)]}$, we are aiming to implement action of the non-unitary operator $e^{\mathcal{L}t}$ on a digital quantum device. Therefore, the following three foundational building blocks are required. First we need to find an approximation of $e^{\mathcal{L}t}$ that can be directly implemented on a quantum computer. This can be achieved by combining a unitary time-evolution of the dilated jump operator $\tilde{L}$ from Eq.~\ref{eqn:L_tilde_definition} that block encodes $L$, with non-unitary measurement operations on the ancilla qubit that is used for the block encoding. Then, to avoid the expensive time evolution under $\tilde{L}$, we need to truncate and discretize the integral from Eq.~\ref{eqn:jumperator_L_definition}, which can be done by means of the trapezoidal rule. Lastly, to implement the unitary evolution of the resulting sum of Hermitian operators, where each constituent contributes to the discretized integral, we make use of a Trotter splitting.

In the following Subsections~\ref{sct:Approximating the Dissipative Part} -- \ref{sct:Trotter Splitting}, we will briefly outline the details of each step and the approximation error with respect to Eq.~\ref{eqn:ground_state_lindblad_equation_Appendix}. Then, we sum up the errors of the three building blocks to verify the functionality of the scheme and thereby estimate the implementation cost. This is presented together with the proof of Proposition~\ref{prop:Linear_scaling} in Section~\ref{sct:Error Analysis}. Lastly, the discrete-time implementation which is optimized for ground-state preparation is summarized in Section~\ref{sct:Discrete-Time Dynamics}. A more thorough analysis of the presented algorithms may be found in the original paper~\cite{Ding2023-rv}.

\subsection{Approximating the Dissipative Part}
\label{sct:Approximating the Dissipative Part}

As outlined above, this subsection aims to find an approximation to $e^{\mathcal{L}t}$ that can be implemented on a digital quantum computer. First, we state again the dilated Hermitian jump operator that block-encodes $L$:
\begin{equation*}
\begin{aligned}
    \tilde{L} = \begin{pmatrix} 0 & L^\dag \\ L & 0 \end{pmatrix}.
\end{aligned}
\end{equation*}
Further, let ${\Tr_a \left( \sum_{i,j = 0}^1 \ket{i}_a \bra{j}_a \otimes \rho_{ij}\right) = \sum_{i=0}^1 \rho_{ii}}$ denote the operation that traces out the subspace of the ancilla qubit, which is used for the block-encoding. To approximate the Lindblad dynamics as governed by Eq.~\ref{eqn:ground_state_lindblad_equation_Appendix}, we need to repeatedly employ a dilated Hamiltonian simulation of $e^{-i\tilde{L}t}$, then trace out the ancilla qubit and reset it to the state $\ket{0}_a$. The \mbox{$m$-th} iteration of the algorithm, propagating the system by a time step $\sqrt{\tau}$, is thus defined as:
\begin{equation}
\label{eqn:sigma_lindblad_approximation_Appendix}
\begin{aligned}
    \sigma(m\tau) &:= \sigma_m(\tau) \\
    &:= \Tr_a \left( e^{-i\tilde{L}\sqrt{\tau}} [\ket{0}_a \bra{0}_a \otimes \sigma_{m-1}(\tau)] e^{i \tilde{L} \sqrt{\tau}}\right),
\end{aligned}
\end{equation}
followed by resetting the ancilla to $\ket{0}_a$. The error ${\norm{\rho(t) - \sigma(t)}_1}$ of the simulation scheme, where $\rho(t)$ is defined via Eq.~\ref{eqn:ground_state_lindblad_equation_Appendix}, can be estimated as follows. Let ${Y_{3,\tilde{L}} = 1-i\tilde{L} \sqrt{t} + \tilde{L}^2 t/2 + \tilde{L}^3 t^{3/2}/6}$ denote the Taylor expansion of ${e^{-i\tilde{L}\sqrt{t}}}$ truncated at third order. We can then use the triangle inequality to obtain:
\begin{equation}
\begin{aligned}
\label{eqn:approximating_dissipative_part_error}
    &\norm{\rho(t) - \sigma(t)}_1 \leq \norm{\rho(t) - (1+\mathcal{L} t) [\rho(0)]}_1 \\
    & + \norm{(1+\mathcal{L}t) [\rho(0)] - \Tr_a \left( Y_{3,\tilde{L}} [\ket{0}_a \bra{0}_a \otimes \rho(0)] Y_{3,\tilde{L}}^\dag\right)}_1 \\
    & + \norm{\Tr_a \left( Y_{3,\tilde{L}} [\ket{0}_a \bra{0}_a \otimes \rho(0)] Y_{3,\tilde{L}}^\dag \right) - \sigma(t)}_1 \\
    &= \mathcal{O} (\norm{L}^4 t^2).
\end{aligned}
\end{equation}
The first term was evaluated by employing the Taylor expansion
\begin{equation*}
    \rho(t) = e^{\mathcal{L} t} [\rho(0)] = (1+\mathcal{L}t)[\rho(0)] + \mathcal{O}\left(\norm{\mathcal{L}}^2 t^2\right).
\end{equation*}
We finally use ${\norm{\mathcal{L}} = \mathcal{O}\left(\norm{L}^2\right)}$ to arrive at the stated error bound. A detailed proof can be found in Ref.~\cite[Appendix B]{cleve_et_al:LIPIcs.ICALP.2017.17}. By a direct calculation, it becomes apparent that the second and third terms exhibit identical scaling. Therefore, note that all factors with a fractional power of the form ${t^{1/2}}$ vanish due to the partial trace operation, since $\Tr_a \tilde{L} = 0$.

\subsection{Truncation and Discretization of the Integral}
\label{sct:Truncation and Discretization of the Integral}

To efficiently encode $\sigma(\tau)$ from Eq.~\ref{eqn:sigma_lindblad_approximation_Appendix}, we need to find an approximation of $e^{-i\tilde{L}t}$. Ref.~\cite{Ding2023-rv} suggests truncating and discretizing the integral from Eq.~\ref{eqn:jumperator_L_definition}. We will see later (Section~\ref{sct:Trotter Splitting}), that the resulting evolution can be implemented directly via Trotterization using only a single ancilla qubit. We want to emphasize that this step is essential to avoid the expensive evolution under a block-encoding of $L$.

According to Ref.~\cite[Lemma 8]{Ding2023-rv}, the filter function ${\gamma(s)}$ decays super-polynomially as ${|s| \rightarrow \infty}$ if we choose ${\hat{\gamma}(\omega)}$ in such a way that it fulfills the following three assumptions:
\begin{equation}
\begin{aligned}
\label{eqn:conditions_on_gamma}
    &0 \leq \hat{\gamma}(\omega) \leq 1\;, \\
    \hat{\gamma}(\omega) = &\begin{cases} \Omega(1) & \text{, if $\omega = -\Delta$} \\ 0 & \text{, if $\omega \notin [-M_{\omega}, 0]$}
    \end{cases}, \\
    &\abs{\frac{\textrm{d}^Z \hat{\gamma}(\omega)}{\textrm{d}\omega^Z}} \leq C \Delta^{-Z}.
\end{aligned}
\end{equation}
\begin{figure}[t]
  \centering
\subfloat[Ideal filter function\label{fig:gamma_w_and_s_a}]{
 \includegraphics{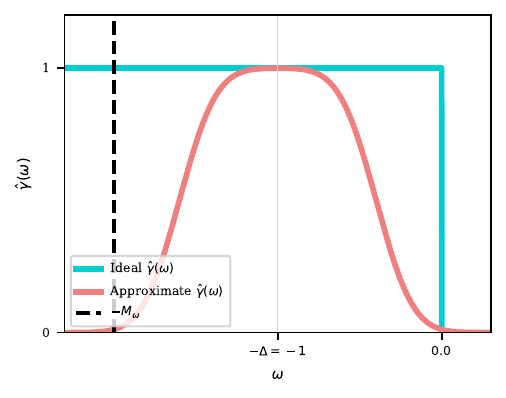}
} \\
\subfloat[Approximate filter function\label{fig:gamma_w_and_s_b}]{
  \includegraphics{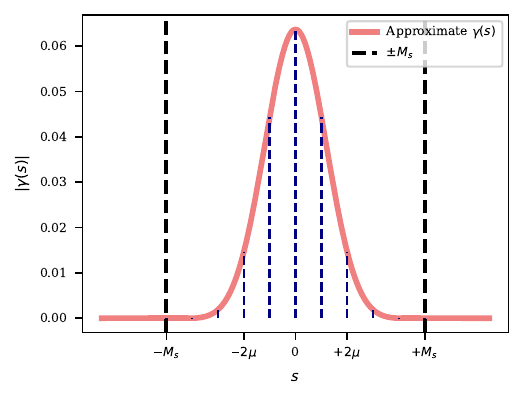}
}
  \caption{\justifying An ideal and an approximate (Eq.~\ref{eqn:gamma_of_omega}) filter function $\hat{\gamma}(\omega)$, suitable for ground-state preparation via Lindbladians for the Grover Hamiltonian with spectral gap $\Delta=1$ is shown in (a). Its inverse Fourier transform $\gamma(s)$ from Eq.~\ref{eqn:gamma_of_s} is depicted in (b). The vertical lines at ${s = \pm M_s}$ indicate where the integral is truncated and the thin lines illustrate an exemplary discretization with a step width of $\mu$.}
  \label{fig:gamma_w_and_s}
\end{figure}
Here, $C, Z > 0$, and $M_{\omega}$ is a hyperparameter of the algorithm. ${\Delta := \lambda_1 - \lambda_0}$ is the spectral gap of the Hamiltonian $H$, denoting the difference between the ground state and first excited state. When $\hat{\gamma}(\omega)$ is zero, it suppresses all transitions; when it is one, it allows all transitions. This explains the first condition. The second condition requires the value of the filter function to be high within a specific range to enable transitions between the ground state and the first excited eigenstates, and zero outside the range $[-M_\omega,0]$. Here, $M_\omega$ must be chosen carefully to both facilitate the transition of amplitude from eigenstates $\lambda_j$ to lower-energy eigenstates $\lambda_i$ with ${\lambda_i - \lambda_j \leq -M_{\omega}}$ and to keep the cost of implementation by truncating the integral low. The last condition in Eq.~\ref{eqn:conditions_on_gamma} helps control the error of the trapezoidal rule. Figure~\ref{fig:gamma_w_and_s} provides an example of both an ideal filter function and an approximate one suitable for ground-state preparation via Lindbladians for the Grover Hamiltonian with spectral gap $\Delta=1$. For our analytical calculations in Section~\ref{sct:Results for Grover's Search}, we assumed a perfect filter function and for the numerical simulations in Figure~\ref{fig:Short_Range_Comparison}, we chose:
\begin{equation}
\label{eqn:gamma_of_omega}
    \hat{\gamma}(\omega)=\frac{1}{2} \left(\erf\left(4\,\omega + 4.8\right) - \erf\left(4\,\omega+3.2\right)\right),
\end{equation}
where $\erf(\omega) = \frac{2}{\sqrt{\pi}}\int_0^\omega e^{-x^2}$ denotes the error function. Its inverse Fourier transform ${\gamma(s)=\int_{\mathbb{R}} \hat{\gamma}(\omega)e^{-i\omega s} \mathrm{d} \omega}$ is given by
\begin{equation}
\label{eqn:gamma_of_s}
    \gamma(s)=\frac{e^{s^2/16}e^{1.2 \,i s} + e^{s^2 / 16} e^{0.8\, i s}}{2 \pi i s}.
\end{equation}
Note that this function is only suitable for the ETH (Section~\ref{sct:Convergence under the Eigenstate Thermalization Hypothesis}) and long-range (Section~\ref{sct:A Single Projector}) approaches. In the short-range regime (Section~\ref{sct:Single Bit-Flip Operators}), we require $\hat{\gamma}(0)=1$.

If we now restrict the integration range of the integral in Eq.~\ref{eqn:jumperator_L_definition} to ${[-M_s, M_s]}$, and use a uniform grid ${\{\ell \mu \mid \ell=-M_\mu, -M_\mu + 1, ..., M_\mu\}}$, and a step width of ${\mu=M_s/M_\mu}$, we arrive at the discrete version of $L$:
\begin{equation}
\label{eqn:discretized_integral_jump_operator_definition}
    L_s := \sum_{\ell=-M_\mu}^{M_\mu} \gamma(\ell \mu) A(\ell \mu) \mu.
\end{equation}
According to \mbox{Ref.~\cite[Lemma 11]{Ding2023-rv}}, when given any constant $Z>0$ and error $\epsilon_s > 0$, if we choose
\begin{align}
\label{eqn:M_s_tau_s}
\begin{split}
    M_s &= \Omega\left(\frac{1}{\Delta} \left(\frac{C M_{\omega} \norm{A}}{\Delta \epsilon_s}\right)^{\frac{1}{Z-1}}\right), \\
    \mu &= \mathcal{O} \left(\left(||H|| + M_{\omega} + \log\left(\frac{C M_{\omega} \norm{A}}{\Delta \epsilon_s}\right)\right)^{-1}\right),
\end{split}
\end{align}
then, we get:
\begin{equation*}
    \norm{L - L_s} = \mathcal{O}(\epsilon_s).
\end{equation*}
By substituting $\tilde{L}$ with its discretized version $\tilde{L}_s$, we arrive at a new approximation for the density matrix: ${\kappa(\tau) := \Tr_a \left( e^{-i\tilde{L}_s\sqrt{\tau}} [\ket{0}_a \bra{0}_a \otimes \rho(0)] e^{i \tilde{L}_s \sqrt{\tau}}\right)}$. Note that the triangle inequality implies
\begin{equation*}
\norm{\rho(\tau) - \kappa(\tau)}_1 \leq \norm{\rho(\tau) - \sigma(\tau)}_1 + \norm{\sigma(\tau) - \kappa(\tau)}_1,
\end{equation*}
where $\sigma(\tau)$ denotes the approximation of the density matrix $\rho(\tau)$ from Eq.~\ref{eqn:sigma_lindblad_approximation_Appendix}. The first term was estimated above in Eq.~\ref{eqn:approximating_dissipative_part_error}. Therefore, our next task is to estimate the second term, leading to:
\begin{equation}
\label{eqn:integral_truncation_and_discretization_error}
    \norm{\sigma(\tau) - \kappa(\tau)}_1 = \mathcal{O}\left( \tnorm{\tilde{L} - \tilde{L}_s} \left(\tnorm{\tilde{L}_s} + \tnorm{\tilde{L}}\right) \tau \right) = \mathcal{O}\left(\epsilon\right).
\end{equation}
In the second equality, we let $Z\rightarrow\infty$ and use the fact that $\mu$ has a logarithmic dependency in $\epsilon_s$, such that we can set ${\epsilon_s = \mathcal{O}\left(\frac{\epsilon}{\left(\tnorm{\tilde{L}_s} + \tnorm{\tilde{L}}\right)\tau}\right)}$. For detailed proofs, we refer the reader to \mbox{Ref.~\cite[Proposition 17]{Ding2023-rv}}.
At this stage, the dilation takes the following form:
\begin{widetext}
\begin{equation*}
    \tilde{L} \approx \tilde{L}_s = \begin{pmatrix}
        0 & \left(\sum_{\ell=-M_\mu}^{M_\mu} \gamma(\ell\mu) A(\ell\mu) \mu_{\ell}\right)^\dag \\
        \sum_{\ell=-M_\mu}^{M_\mu} \gamma(\ell\mu) A(\ell\mu) \mu_{\ell} & 0
    \end{pmatrix} =: \sum_{\ell=-M_\mu}^{M_\mu} \tilde{H}_{\ell}.
\end{equation*}
\end{widetext}
After defining ${\tilde{\sigma}_{\ell} := \mu_{\ell} \left(\sigma_x \Re \gamma(\ell \mu) + \sigma_y \Im \gamma(\ell \mu)\right)}$, where $\sigma_\alpha$ are the Pauli matrices, we see that each term factorizes as:
\begin{equation*}
    \Tilde{H}_{\ell} = \begin{pmatrix}
        0 & \gamma^*(\ell\mu) A(\ell \mu) \mu_{\ell} \\
        \gamma(\ell\mu) A(\ell \mu) \mu_{\ell} & 0
    \end{pmatrix} = \tilde{\sigma}_{\ell} \otimes A(\ell \mu).
\end{equation*}
\subsection{Trotter Splitting}
\label{sct:Trotter Splitting}

The next step consists of the Trotterization of the Hamiltonian evolution ${e^{-i\tilde{L}_s \sqrt{\tau}} = e^{-i\sum_{\ell} \tilde{H}_{\ell} \sqrt{\tau}}}$. Ref.~\cite{Ding2023-rv} chooses a second-order Trotter scheme to find the right balance between efficiency and accuracy. However, we argue that when Trotterizing up to \mbox{$p$-th} order one can achieve a significantly better scaling for certain problem structures. We will come back to this statement in the following subsection. A \mbox{$p$-th} order Trotterized evolution is of the form~\cite{Childs2021}:
\begin{equation}
\label{eqn:unitary_time_evolution_after_trotterization}
    W(\tau) := \prod_{\tilde{\nu}=1}^{\tilde{N}} \prod_{\tilde{\gamma}=1}^{\tilde{\Gamma}} e^{-i b_{(\tilde{\nu},\tilde{\gamma})} \tilde{H}_{\pi_{\tilde{\nu}}(\tilde{\gamma})} \sqrt{\tau}}.
\end{equation}
The coefficients ${b_{(\tilde{\nu},\tilde{\gamma})}}$ are real numbers. For the Suzuki formula, ${\tilde{N}=2 \cdot 5^{p/2-1}}$ denotes the number of stages for even $p$. The permutation $\pi_{\tilde{\nu}}$ controls the ordering of operator summands within stage $\tilde{\nu}$. Hence, the new approximation $\chi(\tau)$ of the density matrix $\kappa(\tau)$ reads:
\begin{equation}
\label{eqn:final_approximated_density_matrix_chi}
    \chi(\tau):=\Tr_a\left(W(\tau)[\ket{0}_a\bra{0}_a \otimes \rho(0)]\left(W(\tau)\right)^{\dag}\right).
\end{equation}
Therefore, the circuit comprises implementations of alternating Hamiltonian evolutions and controlled-$A$ gates:
\begin{equation}
\label{eqn:what_has_to_be_implemented}
\begin{aligned}
    &e^{-i b_{\ell} \tilde{H}_{\ell} \sqrt{\tau}} \\
    &= \underbrace{\left(I \otimes e^{iH\ell\mu}\right)}_{\text{Ham. evol.}} \underbrace{\exp\left(-i b_{\ell} (\tilde{\sigma}_{\ell} \otimes A) \sqrt{\tau}\right)}_{\text{controlled-A}} \underbrace{\left(I \otimes e^{-iH\ell\mu}\right)}_{\text{Ham. evol.}}.
    \end{aligned}
\end{equation}
A high-level circuit diagram is provided in Figure \ref{fig:circuit_illustration}. Ref.~\cite{Ding2023-rv} gives further details on how to optimize ${\chi(\tau)}$ by canceling out back-and-forth Hamiltonian evolution, reducing the Hamiltonian simulation time of each Trotter term from ${\ell \mu}$ to $\mu$. We saw in Section~\ref{sct:Single Bit-Flip Operators} that $A$ can be as easy as a sum of single Pauli operators. Hence, these circuits can be way simpler than a time-evolution under a block-encoding of $\tilde{L}$.

We again make use of the triangle inequality to estimate the overall error of approximating $\rho(\tau)$ by $\chi(\tau)$, which leads to:
\begin{equation*}
\begin{aligned}
    \norm{\rho(\tau) - \chi(\tau)}_1 &\leq \norm{\rho(\tau) - \sigma(\tau)}_1 \\
    &+ \norm{\sigma(\tau) - \kappa(\tau)}_1
+ \norm{\kappa(\tau) - \chi(\tau)}_1.
\end{aligned}
\end{equation*}
The first two terms were estimated in Eq.~\ref{eqn:approximating_dissipative_part_error} and~\ref{eqn:integral_truncation_and_discretization_error}. The error in trace distance of the \mbox{$p$-th} order Trotterization $\chi(\tau)$ to the density matrix $\kappa(\tau)$ after discretizing the integral reads:
\begin{figure}[t]
  \centering
  %% Creator: Inkscape 1.2.2 (732a01da63, 2022-12-09), www.inkscape.org
%% Accompanies image file 'circuit_illustration.pdf' (pdf, eps, ps)
\begingroup%
  \makeatletter%
  \providecommand\color[2][]{%
    \errmessage{(Inkscape) Color is used for the text in Inkscape, but the package 'color.sty' is not loaded}%
    \renewcommand\color[2][]{}%
  }%
  \providecommand\transparent[1]{%
    \errmessage{(Inkscape) Transparency is used (non-zero) for the text in Inkscape, but the package 'transparent.sty' is not loaded}%
    \renewcommand\transparent[1]{}%
  }%
  \providecommand\rotatebox[2]{#2}%
  \newcommand*\fsize{\dimexpr\f@size pt\relax}%
  \newcommand*\lineheight[1]{\fontsize{\fsize}{#1\fsize}\selectfont}%
  \ifx\svgwidth\undefined%
    \setlength{\unitlength}{245.99996948bp}%
    \ifx\svgscale\undefined%
      \relax%
    \else%
      \setlength{\unitlength}{\unitlength * \real{\svgscale}}%
    \fi%
  \else%
    \setlength{\unitlength}{\svgwidth}%
  \fi%
  \global\let\svgwidth\undefined%
  \global\let\svgscale\undefined%
  \makeatother%
  \begin{picture}(1,0.55779267)%
    \lineheight{1}%
    \setlength\tabcolsep{0pt}%
    \put(0,0){\includegraphics[width=\unitlength,page=1]{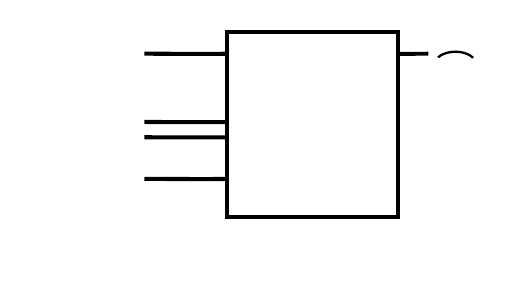}}%
    \put(0.35135851,0.26245415){\color[rgb]{0,0,0}\makebox(0,0)[t]{\lineheight{0.30000001}\smash{\begin{tabular}[t]{c}\textbf{.}\\\textbf{.}\\\textbf{.}\end{tabular}}}}%
    \put(0,0){\includegraphics[width=\unitlength,page=2]{circuit_illustration.pdf}}%
    \put(0.88552634,0.26295152){\color[rgb]{0,0,0}\makebox(0,0)[t]{\lineheight{0.30000001}\smash{\begin{tabular}[t]{c}\textbf{.}\\\textbf{.}\\\textbf{.}\end{tabular}}}}%
    \put(0.70274196,0.01275725){\color[rgb]{0,0.80784314,0.81960784}\makebox(0,0)[t]{\lineheight{1.25}\smash{\begin{tabular}[t]{c}\textbf{Repeat}\end{tabular}}}}%
    \put(0,0){\includegraphics[width=\unitlength,page=3]{circuit_illustration.pdf}}%
    \put(0.6097935,0.29988247){\color[rgb]{0,0,0}\makebox(0,0)[t]{\lineheight{1.25}\smash{\begin{tabular}[t]{c}$W(\tau)$\end{tabular}}}}%
    \put(0.16561963,0.24885395){\color[rgb]{0,0,0}\makebox(0,0)[t]{\lineheight{1.25}\smash{\begin{tabular}[t]{c}$\ket{\Psi}$\end{tabular}}}}%
    \put(0.1686676,0.4376164){\color[rgb]{0,0,0}\makebox(0,0)[t]{\lineheight{1.25}\smash{\begin{tabular}[t]{c}$\ket{0}$\end{tabular}}}}%
  \end{picture}%
\endgroup%

  \caption{\justifying A high-level circuit diagram for the presented implementation of the Lindblad dynamics. Note that the outcome of the measurement on the ancilla qubit can be discarded.}\label{fig:circuit_illustration}
\end{figure}
\begin{equation}
\begin{aligned}
\label{eqn:error_of_trotterization}
    &\norm{\kappa(\tau) - \chi(\tau)}_1 = \mathcal{O}\left(\left(\tau^{1/2} \sum_{\ell} \tnorm{\tilde{H}_{\ell}}\right)^{p+2}\right) \\ &= \mathcal{O}\left(\left(\tau^{1/2} \norm{A} \mu\right)^{p+2} \left(\sum_{\ell} \abs{\gamma(\ell\mu)}\right)^{p+2}\right) \\
    &= \mathcal{O}\left(\tau^{p/2+1} \norm{A}^{p+2}\right).
\end{aligned}
\end{equation}
A precise derivation of Trotter splitting errors, clarifying the first equality, can be found in \mbox{Ref.~\cite[Lemma 1]{Childs2021}}. Ref.~\cite{Ding2023-rv} extends their findings by noticing that due to the partial trace operation, all terms with fractional powers of the form $\tau^{3/2}$ vanish (since ${\Tr_a \tilde{L}=0}$), and hence, the exponent of the error-term is $p+2$ instead of $p+1$. In the last equality, we make use of \mbox{Ref.~\cite[Lemma 8]{Ding2023-rv}} to neglect ${\left(\sum_{\ell} \abs{\gamma(\ell\mu)}\right)^{p+2}}$. We further assume ${\norm{H} = \mathcal{O}(1)}$ (thus, also ${M_{\omega} = \mathcal{O}(1)}$) and ${\Delta = \mathcal{O}(1)}$, as we investigate the Hamiltonian of an unstructured search space from Eq.~\ref{eqn:grover_hamiltonian} and hence also neglect $\mu$ in the last equality. An analysis that includes these parameters is provided in \mbox{Ref.~\cite[Theorem 13]{Ding2023-rv}}.

\subsection{Error Analysis}
\label{sct:Error Analysis}

In this section, we sum up the error terms of the approximations in Sections~\ref{sct:Approximating the Dissipative Part} -- \ref{sct:Trotter Splitting} and elaborate on the choice of the order of Trotterization, thereby refining the original implementation of Ref.~\cite{Ding2023-rv}.

Given a stopping time ${T>0}$ and accuracy ${\epsilon > 0}$, the goal is to determine the time step $\tau$, such that the error of simulating Eq.~\ref{eqn:ground_state_lindblad_equation_Appendix} by Eq.~\ref{eqn:unitary_time_evolution_after_trotterization} becomes arbitrarily small: ${\norm{\rho(T) - \chi(T)}_1 \leq \epsilon}$. We choose $M_s$ and $\mu$ as stated in Eq.~\ref{eqn:M_s_tau_s}. Originally, the overall complexity was given in terms of the total Hamiltonian simulation time $T_H$ and the number $N_A$ of controlled-$A$ gates, which arise from Eq.~\ref{eqn:what_has_to_be_implemented}~\cite{Ding2023-rv}. Here, we will only state the former as we expect the complexity to be equivalent in the case of a Grover Hamiltonian, since ${N_A = \mathcal{O}\left(T_{H} \mu^{-1}\right) = \mathcal{O}\left(T_H (\norm{H} + M_\omega)^{-1}\right) = \mathcal{O}(T_H)}$.

Using that the Lindblad dynamics are contractive in trace distance for arbitrary density matrices $\tilde{\rho}_1$ and $\tilde{\rho}_2$, we have~\cite{RUSKAI1994}:
\begin{equation*}
    \norm{e^{\mathcal{L}t} \tilde{\rho}_1 - e^{\mathcal{L}t} \tilde{\rho}_2}_1 \leq \tnorm{\tilde{\rho}_1 - \tilde{\rho}_2}_1.
\end{equation*}
Thus, it suffices to estimate the error after the first time step $\tau$ and then multiply it by the number ${\mathcal{N} = T/\tau}$ of employed steps:
\begin{equation}
\label{eqn:overall_error_sum_and_splitup}
\begin{aligned}
    \norm{\rho(T) - \chi(T)}_1 \leq & \mathcal{N} \norm{\rho(\tau) - \chi(\tau)}_1 \\
    \leq & \mathcal{N} \left( \norm{\rho(\tau) - \sigma(\tau)}_1 + \norm{\sigma(\tau) - \kappa(\tau)}_1 \right. \\
    & \left. + \norm{\kappa(\tau) - \chi(\tau)}_1 \right)
\end{aligned}
\end{equation}
In the second inequality, we split the overall error into the Lindblad simulation error, the error from numerical integration, and the error from Trotter splitting by using the triangle inequality. Plugging inequalities \ref{eqn:approximating_dissipative_part_error}, \ref{eqn:integral_truncation_and_discretization_error}, and \ref{eqn:error_of_trotterization} into Eq.~\ref{eqn:overall_error_sum_and_splitup}, we get:
\begin{equation*}
\begin{aligned}
    &\norm{\rho(T) - \chi(T)}_1 \\
    &= \mathcal{O}\left(\norm{L}^4 \tau T\right) + \mathcal{O}\left(\epsilon\right) + \mathcal{O}\left(\tau^{p/2} T \norm{A}^{p+2}\right) = \mathcal{O}\left(\epsilon\right).
\end{aligned}
\end{equation*}
The second equality holds if we choose $\tau$ as:
\begin{equation*}
    \tau = \mathcal{O}\left(\min \left(
    \norm{L}^{-4} T^{-1} \epsilon, 
    \norm{A}^{-2-4/p} \, T^{-2/p} \, \epsilon^{2/p}\right)\right).
\end{equation*}
The overall complexity in terms of the total Hamiltonian simulation time $T_H$ can now be given by:
\begin{equation}
\begin{aligned}
\label{eqn:continuous_time_overall_complexity_lindblad_Appendix}
    T_H &= \Theta\left(\frac{T}{\tau}\right) \\
    &= \tilde{\Theta}\left(\max\left(\norm{L}^4 T^2 \epsilon^{-1}, \norm{A}^{2+o(1)} T^{2+o(1)} \epsilon^{-o(1)}\right)\right).
\end{aligned}
\end{equation}
We now come to prove Proposition~\ref{prop:Linear_scaling}, which states that this outcome can be refined to a linear dependency in $T$ while still following the Lindblad dynamics when preparing the ground state of the Hamiltonian of an unstructured search space.
\begin{namedtheorem}
    To find the ground state of an unstructured search space with $M$ marked elements, the overall complexity of implementation in terms of the total Hamiltonian simulation time $T_H$ of the approximated continuous-time dynamics scales linearly in the mixing time. This linear scaling is achieved by choosing the coupling operator as $A:=J_N/\sqrt{(N-M)M}$, where $J_N$ denotes the all-one matrix of size $N\times N$ and the filter function $\hat{\gamma}(\omega)$ is defined to be zero for $\omega\geq0$ and one otherwise.
\end{namedtheorem}
\begin{proof}
First, we show that the quantity $\norm{A}^2 T_{\mathrm{mix}}$ is constant. When defining ${\tilde{A} := \eta A}$, we can see that ${\tilde{\mathcal{L}} = \eta^2 \mathcal{L}}$ for any constant $\eta$. Therefore, the eigenvalues of $\mathcal{L}$ scale proportionally to $\eta^2$. Hence, we have $\tilde{T}_{\mathrm{mix}}=\eta^{-2} T_{\mathrm{mix}}$, which leads to:
\begin{equation*}
    \tnorm{\tilde{A}}^2 \tilde{T}_{\mathrm{mix}} = \norm{A}^2 T_{\mathrm{mix}} = \textrm{const}.
\end{equation*}
Next, we evaluate the mixing time for the Hamiltonian $H=-\sum_{g_i \in \mathcal{G}} \ket{g_i} \bra{g_i}$, where $\mathcal{G}$ denotes the set of all marked elements with ${\abs{G}=M}$. Let $A:=\ket{s}\bra{s}$, where ${\ket{s}=\frac{1}{\sqrt{N}}\sum_i\ket{i}}$ is the uniform superposition state. When defining the variables~${z_{g} := \sum_{g_i \in \mathcal{G}}\sum_{g_j \in \mathcal{G}} \rho_{g_{i}g_{j}}}$ and~${z_{e} := \sum_{a \notin \mathcal{G}}\sum_{b \notin \mathcal{G}} \rho_{ab}}$, we get from the LME~\ref{eqn:lindblad_equation_in_indices}:
\begin{equation*}
\begin{aligned}
    \dot{z}_{g}&=\eta^2 M z_{e} \\
    \dot{z}_{e} &= -(N-M) M \eta^2 z_{e}.
\end{aligned}
\end{equation*}
Hence, the mixing time is $T_{\mathrm{mix}}=1/((N-M) M \eta^2)$. When we choose $\eta=1/\sqrt{(N-M)M}$, we get $T_{\mathrm{mix}}=1$. Finally, we want to evaluate $\norm{L}$, where $L$ is defined via Eq.~\ref{eqn:jumperator_L_definition} as:
\begin{equation*}
\begin{aligned}
    L=\sum_{ij}\hat{\gamma}_{ij} A_{ij} \ket{i} \bra{j} = \frac{1}{\sqrt{(N-M)M}} \sum_{\substack{i\in \mathcal{G} \\ j \notin \mathcal{G}}} \ket{i}\bra{j}.
\end{aligned}
\end{equation*}
The Frobenius norm of the jump operator is ${\norm{L}_F=\sqrt{\sum_{ij}\big | L_{ij}\big | ^2}}=1$. By using the identity $\norm{L}\leq \norm{L}_F \leq \sqrt{r} \norm{L}$, where $r:=\mathrm{rank}(L)=1$, we see that $\norm{L}=1$. Given that $\norm{A}^2 T_{\mathrm{mix}}$ is constant, Equation~\ref{eqn:continuous_time_overall_complexity_lindblad_Appendix} implies a linear scaling of the overall complexity with the mixing time.
\end{proof}
As we already observed in Section~\ref{sct:Single Bit-Flip Operators}, Proposition~\ref{prop:Linear_scaling} does not hold for general approaches. Nevertheless, we believe that this result can be extended to other interesting systems. In a broader sense, if we assume that ${\int_{-\infty}^{\infty}\gamma(s)\textrm{d}s} = \mathcal{O}(1)$, then:
\begin{equation*}
    \norm{L} \leq \norm{A} \int_{-\infty}^{\infty} \abs{\gamma(s)}\textrm{d}s = \mathcal{O}\left(\norm{A}\right).
\end{equation*}
However, the ratio of $\norm{A}/\norm{L}$ can still be exponentially large in the system size $n$, as seen in the case of the Grover Hamiltonian.

\subsection{Discrete-Time Dynamics}
\label{sct:Discrete-Time Dynamics}

As already noted in Ref.~\cite{Ding2023-rv}, our objective might not necessarily be to precisely replicate the Lindblad dynamics but rather to prepare the ground state. Thus, we can tailor the algorithm to ground-state preparation by recognizing that the steady-state argument in Eq.~\ref{eqn:steady_state_argument} remains valid for $\sigma(t)$ from Eq.~\ref{eqn:sigma_lindblad_approximation_Appendix}:
\begin{equation*}
\begin{aligned}
    &\tilde{L} \left(\ket{0}_a \bra{0}_a \otimes \rho_g\right) = \begin{pmatrix}
        0 \\ L \rho_g
    \end{pmatrix} = \begin{pmatrix}
        0 \\ 0
    \end{pmatrix} \\
    \Longleftrightarrow \; &\sigma(t)[\rho_g]=\Tr_a \left( e^{-i\tilde{L}\sqrt{t}} [\ket{0}_a \bra{0}_a \otimes \rho_g] e^{i \tilde{L} \sqrt{t}}\right) = \rho_g.
\end{aligned}
\end{equation*}
This allows us to choose a larger time step and ignore the error from Eq.~\ref{eqn:approximating_dissipative_part_error}. Having shown in Eq.~\ref{eqn:integral_truncation_and_discretization_error} that the error from truncating and discretizing the integral is insignificant for the right choice of $M_{s}$ and $\mu$, we observe that the main limitation once more lies in achieving precise Trotterization:
\begin{equation*}
    \norm{\sigma(\tau) - \chi(\tau)}_1 = \mathcal{O} \left(\norm{A}^{p+2} \tau^{p/2+1} \right).
\end{equation*}
To achieve an accuracy of ${\epsilon>0}$, set:
\begin{equation*}
    \tau = \mathcal{O}\left(\norm{A}^{-2-4/p} \, T^{-2/p} \, \epsilon^{2/p}\right).
\end{equation*}
We see, that the discrete-time dynamics have the disadvantage of not following the exact Lindblad dynamics, but improve the cost of ground-state preparation to a linear dependency in $T$ and a negligible dependency in $\epsilon$ for a high Trotter order $p$:
\begin{equation*}
    T_H = \Theta\left(\frac{T}{\tau}\right) = \tilde{\Theta}\left(T^{1+o(1)} \norm{A}^{2+o(1)} \epsilon^{-o(1)}\right).
\end{equation*}

\section{Quantum Continuous-Time Short-Range Convergence}
\label{sct:Quantum Continuous-Time Short-Range Convergence}

In this section, we want to prove the convergence to the ground state for the short-range transition ansatz from Section~\ref{sct:Single Bit-Flip Operators}.

\begin{lemma}
\label{lma:short_range_convergence_proof}
    For the short-range transition ansatz with ${A=1/n  \sum_i X_{i}}$, ${\hat{\gamma}(0) = 1}$, and the Hamiltonian in Eq.~\ref{eqn:grover_hamiltonian}, the dynamics of the LME converge to the ground state, when starting in $\rho(0)=\ket{s}\bra{s}$.
\end{lemma}

\begin{proof}
In the following, we will employ Lemma~\ref{lma:Convergence}. First, we determine all kernels of $L$ by applying it to a state ${\ket{\Psi}:=\sum_{\alpha=0}^n \sum_{m \in \mathcal{H}_\alpha} c_m \ket{m}} \in \mathbb{C}^N$:
\begin{align*}
    L \ket{\Psi} &= \sum_{\alpha=1}^{n-1} \sum_{j \in \mathcal{H}_\alpha} \sum_{i \in \mathcal{H}_{\alpha+1} \oplus \mathcal{H}_{\alpha-1}} c_j \ket{i} + \sum_{j \in \mathcal{H}_n} \sum_{i \in \mathcal{H}_{n-1}} c_j \ket{i} \\
    &= \sum_{i \in \mathcal{H}_0} \sum_{j \in \mathcal{H}_1} c_j \ket{i} + \sum_{i \in \mathcal{H}_1} \sum_{j \in \mathcal{H}_2} c_j \ket{i} \\
    &+ \sum_{\alpha=2}^{n-1} \sum_{i \in \mathcal{H}_\alpha} \sum_{j \in \mathcal{H}_{\alpha+1} \oplus \mathcal{H}_{\alpha-1}} c_j \ket{i} + \sum_{i \in \mathcal{H}_n} \sum_{j \in \mathcal{H}_{n-1}} c_j \ket{i}.
\end{align*}
From this, we can follow that the ground state $\ket{g}$ is a steady state of the evolution, since $L\ket{g}=0$. Nevertheless, there are other subspaces ${\mathcal{H}^{\perp} := \mathrm{span}\{\ket*{\Psi^{\perp}}\}}$ that lie in the kernel of $L$:
\begin{widetext}
\begin{equation*}
\begin{aligned}
    \mathcal{H}^{\perp} &= \mathrm{span}\Bigg\{ \sum_{i \in \mathcal{H}_0} \sum_{j \in \mathcal{H}_1} c_j \ket{i} \mid \sum_{j \in \mathcal{H}_1} c_j = 0 \Bigg\} \oplus \mathrm{span}\Bigg\{ \sum_{i \in \mathcal{H}_1} \sum_{j \in \mathcal{H}_2} c_j \ket{i} \mid \sum_{j \in \mathcal{H}_2} c_j = 0 \Bigg\} \\
    &\oplus \; \mathrm{span}\Bigg\{ \sum_{i \in \mathcal{H}_{\alpha}} \sum_{j \in \mathcal{H}_{\alpha+1} \oplus \mathcal{H}_{\alpha-1}} c_j \ket{i} \mid \sum_{j \in \mathcal{H}_{\alpha+1} \oplus \mathcal{H}_{\alpha-1}} c_j = 0, \; 1 < \alpha < n \Bigg\}
    \oplus \; \mathrm{span}\Bigg\{ \sum_{i \in \mathcal{H}_n} \sum_{j \in \mathcal{H}_{n-1}} c_j \ket{i} \mid \sum_{j \in \mathcal{H}_{n-1}} c_j = 0 \Bigg\}
\end{aligned}
\end{equation*}
\end{widetext}
Note however that when we start in $\rho(0)=\ket{s}\bra{s}$, which has no overlap with $\mathcal{H}^{\perp}$ (since $\Tr_{\mathcal{H}^\perp} \rho(0)=0$), $\rho(t)$ will stay in $\mathcal{H} \setminus \mathcal{H}^{\perp}, \; \forall t\geq 0$, where $\mathcal{H} \in \mathbb{C}^N$ denotes the Hilbert space of the system. This can be understood by realizing that:
\begin{equation*}
\begin{aligned}
    \bra*{\Psi^\perp} L \ket*{s} = 0 \; &\Longrightarrow \; \Tr_{\mathcal{H}^\perp} \left(L \rho(0) L^\dag\right)=0 \\
    \bra*{s} \ket*{\Psi^\perp} = 0 \; &\Longrightarrow \; \Tr_{\mathcal{H}^\perp} \left(L^\dag L \rho(0)\right)=0
\end{aligned}
\end{equation*}
and hence we have $\forall t\geq 0$,
\begin{equation*}
\Tr_{\mathcal{H}^\perp} \left(\mathcal{L}_L[\rho(0)]\right)=0 \; \Longrightarrow \; \Tr_{\mathcal{H}^\perp} \left(e^{\mathcal{L}_L t}[\rho(0)]\right)=0.
\end{equation*}
We can now employ Lemma~\ref{lma:Convergence} in the restricted subspace $\mathcal{H} \setminus \mathcal{H}^{\perp}$. Therefore, what remains to be shown is that there is no subspace $S \subseteq \mathcal{H} \setminus \mathcal{H}^{\perp}$ which is left invariant under $L$. We first demonstrate that $L$ maps states in the subspace ${\mathcal{H}_{\alpha}}$, namely ${\ket{\Psi_{\alpha}}:=\sum_{m \in \mathcal{H}_{\alpha}} c_m \ket{m}}$, only to the adjacent spaces $\mathcal{H}_{\alpha \pm 1}$:
\begin{equation*}
\begin{aligned}
    L \ket{\Psi_{\alpha}} &= \sum_{j \in \mathcal{H}_{\alpha}} \sum_{i \in \mathcal{H}_{\alpha+1} \oplus \mathcal{H}_{\alpha-1}} c_j  A_{ij} \ket{i} \\
    &=
    \begin{cases}
      0 & \text{if $\alpha = 0$} \\
      \sum_{j \in \mathcal{H}_{\alpha}} \sum_{i \in \mathcal{H}_{\alpha-1}} c_j \ket{i} & \text{if $\alpha = n$} \\
      \sum_{j \in \mathcal{H}_{\alpha}} \sum_{i \in \mathcal{H}_{\alpha+1} \oplus \mathcal{H}_{\alpha-1}} c_j \ket{i} & \text{otherwise}.
    \end{cases}
\end{aligned}
\end{equation*}
Next, we demonstrate that when applying $L$ to a state which has overlap with the two adjacent groups ${\mathcal{H}_{\alpha}}$ and ${\mathcal{H}_{\alpha+1}}$ ($\mathcal{H}_{\alpha}$ and $\mathcal{H}_{\alpha \pm 2}$ is equivalent), the resulting state has a non-zero overlap with the subspaces $\mathcal{H}_{\alpha + 2}$ and $\mathcal{H}_{\alpha-1}$. Therefore, define $\ket{\Psi} := \ket{\Psi_{\alpha}} + \ket{\Psi_{\alpha+1}} = \left(\sum_{m \in \mathcal{H}_{\alpha}} c_{m,\alpha} + \sum_{m \in \mathcal{H}_{\alpha+1}} c_{m,\alpha+1}\right) \ket{m}$ and apply $L$:
\begin{equation*}
\begin{aligned}
    L \ket{\Psi} &= \sum_{j \in \mathcal{H}_{\alpha}} \sum_{i \in \mathcal{H}_{\alpha+1} \oplus \mathcal{H}_{\alpha-1}} c_{j\alpha}  A_{ij} \ket{i} \\
    &+ \sum_{j \in \mathcal{H}_{\alpha+1}} \sum_{i \in \mathcal{H}_{\alpha+2} \oplus \mathcal{H}_{\alpha}} c_{j,\alpha+1}  A_{ij}\ket{i}.
\end{aligned}
\end{equation*}
Hence, for an invariant subspace to occur, the amplitudes in the manifolds of $\mathcal{H}_{\alpha+2}$ and $\mathcal{H}_{\alpha-1}$ need to be zero:
\begin{equation*}
\begin{aligned}
    \sum_{j \in \mathcal{H}_{\alpha}} \sum_{i \in \mathcal{H}_{\alpha-1}} c_{j\alpha} A_{ij} &= \alpha \sum_{i \in \mathcal{H}_{\alpha}} c_{i\alpha} \overset{!}{=} 0 \\
    \sum_{j \in \mathcal{H}_{\alpha+1}} \sum_{i \in \mathcal{H}_{\alpha+2}} c_{j,\alpha+1} A_{ij} &= (n-\alpha-1) \sum_{i \in \mathcal{H}_{\alpha+1}} c_{i,\alpha+1} \overset{!}{=} 0.
\end{aligned}
\end{equation*}
From this, we can easily follow, that for two adjacent groups $\mathcal{H}_{\alpha}$ and $\mathcal{H}_{\alpha+1}$, we only yield an invariant subspace if $\ket{\Psi} = 0$. If this was not the case, it would be possible that states in $\mathcal{H}_{\alpha}$ map only to $\mathcal{H}_{\alpha+1}$ and vice versa and thus the manifold $\mathcal{H}_{\alpha} \oplus \mathcal{H}_{\alpha+1}$ could be an invariant subspace. Repeated application of this argument will lead us to arrive at $\ket{g}$ eventually and hence this subspace is not orthogonal to $\ket{g}$. Finally, it is obvious that $\ket{g}$ is an eigenstate of  $H=0$, which is the Hamiltonian of the unitary generator $\mathcal{L}_H=0$.
\end{proof}

\section{Quantum Continuous-Time Short-Range Dynamics}
\label{sct:Quantum Continuous-Time Short-Range Dynamics}

The ODEs for ${A=\eta  \sum_i X_i}$ read:
\begin{equation}
\begin{aligned}
\label{eqn:short_range_ODEs}
    \frac{\dot{z}^{\alpha}_{\alpha'}}{\eta^2} = &(\alpha'+1)  (\alpha+1)  z^{\alpha+1}_{\alpha'+1} \\
    &+ (\alpha'+1)  (n-\alpha+1)  z^{\alpha-1}_{\alpha'+1}  (1-\delta_{\alpha,1}) \\
    &+ (\alpha + 1)  (n - \alpha' + 1)  z^{\alpha + 1}_{\alpha'-1}  (1-\delta_{\alpha',1}) \\
    &+ (n - \alpha + 1)  (n-\alpha'+1)  z_{\alpha'-1}^{\alpha-1}  (1-\delta_{\alpha,1})  (1-\delta_{\alpha',1}) \\
    &-\frac{1}{2}  (n-\alpha + 1)  (n - \alpha + 2)  z_{\alpha'}^{\alpha-2}  (1- \delta_{\alpha-2,0}) \\
    &-\frac{1}{2} (n-\alpha+1)  \alpha  z_{\alpha'}^{\alpha} \\
    &-\frac{1}{2}  (\alpha+1)  (n-\alpha)  z_{\alpha'}^{\alpha}  (1-\delta_{\alpha,0}) \\
    &-\frac{1}{2}  ((\alpha'+1)  (n-\alpha')  (1-\delta_{\alpha',0}) + (n-\alpha'+1)  \alpha')  z_{\alpha'}^{\alpha} \\
    &-\frac{1}{2}  (\alpha+1)  (\alpha+2)  z^{\alpha+2}_{\alpha'}  (1 - \delta_{\alpha,0}) \\
    &-\frac{1}{2}  (\alpha'+1)(\alpha+2)  z_{\alpha'+2}^{\alpha}  (1-\delta_{\alpha',0}) \\
    &-\frac{1}{2}  (n-\alpha'+1)(n-\alpha'+2)  z_{\alpha-2}^{\alpha}  (1-\delta_{\alpha'-2,0}).
\end{aligned}
\end{equation}
Here, we defined ${z^{\alpha}_{\alpha'}=\sum_{i \in \mathcal{H}_{\alpha}} \sum_{j \in \mathcal{H}_{\alpha'}} \rho_{ij}}$, and $\delta_{ij}$ is the Kronecker delta. The latter arises due to the asymmetry, introduced by $\hat{\gamma}$. Further, one has to set all terms ${z_{\alpha'}^{\alpha}, z_{\alpha}^{\alpha'} = 0}$ for ${\alpha, \alpha' < 0}$, or ${\alpha,\alpha' > n}$. For the numerical simulations presented in Figure~\ref{fig:short_range_scaling}, the following initial conditions were used:
\begin{equation*}
    z_\alpha^{\alpha'}(0) = \frac{\binom{n}{\alpha}  \binom{n}{\alpha'}}{N}.
\end{equation*}
\section{Discrete-Time Long-Range Approach}
\label{sct:Discrete-Time Long-Range Approach_Appendix}

\begin{proposition}
    For a coupling operator of the form ${A=\ket{s}\bra{s}}$, starting in ${\rho(0) = \ket{s}\bra{s}}$, after $m$ steps of propagating according to Eq.~\ref{eqn:sigma_lindblad_approximation}, we arrive at the density matrix $\sigma_m$ of Eq.~\ref{eqn:single_projector_discrete_time_density_matrix}.
\end{proposition}

\begin{proof}
We make use of the definitions ${\ket{\tilde{s}} := \frac{1}{\sqrt{N-1}} \sum_{j \neq w} \ket{j}}$, ${\tilde{N} := \sqrt{\frac{N-1}{N}}}$, and ${\zeta := \eta \sqrt{N-1} \sqrt{\tau}}$. First, we rewrite $L$ as
\begin{equation*}
    L = \sum_{ij} \hat{\gamma}_{ij} A_{ij} \ket{i} \bra{j} = \eta \sqrt{N-1} \ket{g} \bra{\tilde{s}}.
\end{equation*}
We can then express $\tilde{L}^{\ell}$ for $\ell \in \mathbb{N}_0$ in terms of ${\tilde{L}}$ and ${\tilde{G}:= \begin{pmatrix}
    \ket{\tilde{s}} \bra{\tilde{s}} & 0 \\ 0 & \ket{g}\bra{g}
\end{pmatrix}}$:
\begin{equation*}
\begin{aligned}
    \tilde{L}^{2\ell} &= \eta^{2\ell} (N-1)^{\ell} \tilde{G}, \; \ell \in \mathbb{N} \\
    \tilde{L}^{2\ell+1} &= \eta^{2\ell} (N-1)^{\ell} \tilde{L} ,\; \ell \in \mathbb{N}_0.
\end{aligned}
\end{equation*}
Hence, we are able to separate the matrix exponential ${e^{-i\tilde{L}\sqrt{\tau}}}$ into an even and an odd part as follows:
\begin{equation}
\begin{aligned}
\label{eqn:exponential_discrete_single_projector}
    e^{-i\tilde{L}\sqrt{\tau}} &= \sum_{k=0}^\infty \frac{\left(-i \tilde{L} \sqrt{\tau}\right)^k}{k!} \\
    &= I - \tilde{G} + \underbrace{\sum_{k=0}^\infty \frac{\left(-i\right)^{2k}}{(2k)!} \zeta^{2k} \tilde{G}}_{\text{even}} + \underbrace{\sum_{k=0}^\infty \frac{\left(-i\right)^{2k+1}}{(2k+1)!} \zeta^{2k} \tilde{L}}_{\text{odd}} \\ &= I - \tilde{G} + \cos(\zeta) \tilde{G} - \frac{i}{\zeta} \sin(\zeta) \tilde{L}.
\end{aligned}
\end{equation}
${e^{i\tilde{L}\sqrt{\tau}}}$ is simply the complex conjugate of Eq.~\ref{eqn:exponential_discrete_single_projector}. Define the superoperator ${\mathcal{L}_D[\rho]:=\Tr_a \left( e^{-i\tilde{L}\sqrt{\tau}} [\ket{0}_a \bra{0}_a \otimes \rho] e^{i \tilde{L} \sqrt{\tau}}\right)}$. Then, the following can be shown by direct calculation:
\begin{equation*}
\begin{aligned}
    \mathcal{L}_D[\ket{s}\bra{s}] &= \ket{s}\bra{s} + \tilde{N} (\cos(\zeta)-1) (\ket{\tilde{s}}\bra{s} + \ket{s} \bra{\tilde{s}}) \\ &\;\;\;\;+ \tilde{N}^2 (\cos(\zeta)-1)^2 \ket{\tilde{s}}\bra{\tilde{s}} + \tilde{N}^2 \sin(\zeta) \ket{g}\bra{g} \\
    \mathcal{L}_D[\ket{g}\bra{g}] &= \ket{g}\bra{g} \\
    \mathcal{L}_D[\ket{\tilde{s}}\bra{s}] &= \cos(\zeta) \ket{\tilde{s}} \bra{s} + \tilde{N} \cos(\zeta) (\cos(\zeta)-1) \ket{\tilde{s}} \bra{\tilde{s}} \\
    &\;\;\;\;+ \tilde{N} \sin^2(\zeta) \ket{g}\bra{g} \\
    \mathcal{L}_D[\ket{\tilde{s}}\bra{\tilde{s}}] &= \cos^2(\zeta) \ket{\tilde{s}} \bra{\tilde{s}} + \sin^2(\zeta) \ket{g} \bra{g}.
\end{aligned}
\end{equation*}
Note that $\ket{g}$ is a fixed point, as expected from our reasoning in Section~\ref{sct:Continuous-Time Dynamics}. With this in hand, we proceed by proving the formula \ref{eqn:single_projector_discrete_time_density_matrix} by induction. For $m=0$, we obviously have ${\sigma_0 = \ket{s} \bra{s}}$. Then, we perform the step ${m \rightarrow m+1}$:
\begin{widetext}
\begin{equation*}
\begin{aligned}
    \sigma_{m+1}(\tau) = \mathcal{L}_D [\sigma_m] = &\ket{s}\bra{s} + \tilde{N} (\cos(\zeta)-1) \left(\sum_{k=0}^{m-1} \cos^k(\zeta)\right) (1+\cos(\zeta)) (\ket{\tilde{s}}\bra{s} + \ket{s}\bra{\tilde{s}}) \\
    &+ \tilde{N}^2 [(\cos(\zeta)-1)^2-2\cos(\zeta)(\cos(\zeta)-1)(1-\cos^m(\zeta)) \\
    &+ (1-2\cos^m(\zeta)+\cos^{2m}(\zeta)) \cos^2(\zeta)] \ket{\tilde{s}} \bra{\tilde{s}} \\
    &+ \tilde{N}^2 [\sin^2(\zeta) \cos^{2m}(\zeta)+1-\cos^{2m}(\zeta)] \ket{g}\bra{g}.
\end{aligned}
\end{equation*}
\end{widetext}
Here, we made use of the formula for the truncated geometric series:
\begin{equation*}
    \sum_{k=0}^{m-1} \cos^k(\zeta) = \frac{1-\cos^m(\zeta)}{1-\cos(\zeta)}.
\end{equation*}
Note that this formula is only valid for $\cos(\zeta) \neq 1$. However, we assume that $\zeta \neq 2\pi l$ for $l \in \mathbb{Z}$. For $\cos(\zeta)=1$, the sum yields $m$. When we use ${\sin^2(x)=1-\cos^2(x)}$ for the ${\ket{g}\bra{g}}$ term, this reduces to ${\sigma_{m+1}}$ in Eq.~\ref{eqn:single_projector_discrete_time_density_matrix} which completes the proof.
\end{proof}

\section{Influence of the Hamiltonian}
\label{sct:Influence of the Hamiltonian}

Lastly, we want to swiftly analyze the influence of the Hamiltonian evolution via $\mathcal{L}_H$ in Eq.~\ref{eqn:original_lindblad_equation}. We partition the Fock-Liouville space into $z_g := \rho_{gg}$, ${z_{eg}:=\sum_{a\neq g} \rho_{ag}}$, ${z_{ge}:=\sum_{a\neq g} \rho_{ga}}$ and ${z_e := \sum_{a \neq g}\sum_{b \neq g} \rho_{ab}}$. For ${H = -\ket{g} \bra{g}}$, the Von-Neumann Equation reads:
\begin{equation*}
    \dot{\rho}_{ab}=-i \bra{a} [H, \rho] \ket{b} = i (\delta_{ga}  \rho_{gb} - \delta_{bg}  \rho_{ag}).
\end{equation*}
This leads to:
\begin{equation*}
    \dot{z}_g = 0, \,\, \dot{z}_e = 0, \,\, \dot{z}_{eg} = -i \; z_{eg}, \,\, \textrm{and} \,\,
    \dot{z}_{ge} = i \; z_{ge}.
\end{equation*}
From this we can follow that $z_g$ and $z_e$ are not altered when considering the combined unitary-dissipative evolution via ${\mathcal{L} = \mathcal{L}_H + \mathcal{L}_L}$, and thus the populations $\mu_a := \rho_{aa}$ are not affected.

As mentioned earlier, the filter function is usually chosen imperfectly for a realistic implementation, such that ${\hat{\gamma}(0) \neq 0}$, which leads to unwanted effects in, e.g., the long-range approach. The sharper the function, the higher the cost of truncating the integral, as we expect the full width at half maximum (FWHM) of the Fourier transformed function $\gamma(s)$ to be inversely proportional to the original FWHM for the step function $\hat{\gamma}(\omega)$. Numerically, we observe that the algorithm converges to the ground state for an imperfect filter function, fulfilling ${\gamma(0) \neq 0}$, only when incorporating the Hamiltonian evolution $\mathcal{L}_H$. We analyze this behavior by employing a system of ${n=1}$ qubit. We work in the basis of the Hamiltonian and choose $L$ as:
\begin{equation*}
L=
    \begin{pmatrix}
        \epsilon & 1 \\
        \phi & \epsilon
    \end{pmatrix}.
\end{equation*}
This implies the presence of slight upward transitions via $\epsilon$ and $\phi$, caused by the imperfections in the filter function. When solving the Lindblad equation and deriving the stationary state without the Hamiltonian contributions, we get: \\
\begin{widetext}
\begin{equation*}
\begin{aligned}
\rho_{\mathrm{stat}}=
    \begin{pmatrix}
        1 & 0 \\
        0 & 0
    \end{pmatrix}
&+ \epsilon
    \begin{pmatrix}
        0 & -\left(\frac{1}{1+4i} + \phi \frac{1}{1-4i}\right) \\
        -\left(\frac{1}{1-4i} + \phi \frac{1}{1+4i}\right) & 0
    \end{pmatrix}
- \left(\frac{\epsilon^2}{1+16} + \phi^2\right) Z + \mathcal{O}\left(\epsilon^2 + \phi^2\right).
\end{aligned}
\end{equation*}
\end{widetext}
Thus, the stationary state is not exactly the ground state, whereas when incorporating $\mathcal{L}_H$:
\begin{equation*}
\rho_{\mathrm{stat}}=
    \begin{pmatrix}
        1 & 0 \\
        0 & 0
    \end{pmatrix}.
\end{equation*}

\end{sloppypar}

\bibliography{main}

\end{document}